%% file: branching-processes.tex
\pdfoutput=1
\documentclass[a4paper,UKenglish,cleveref, autoref, thm-restate]{lipics-v2021}

\usepackage{amsmath,amsthm,amsfonts,amssymb,mathtools,stmaryrd,tikz}
\usepackage[all]{xy}
\usepackage{wasysym}
\usetikzlibrary{arrows}
\tikzset{every picture/.style={thick,>=angle 60}}
\tikzset{MDPrand/.style={draw,circle,minimum size=11*1.5,inner sep=1mm}}

\crefname{algorithm}{Algorithm}{Algorithms}
\crefname{appendix}{Appendix}{Appendices}
\crefname{corollary}{Corollary}{Corollaries}
\crefname{equation}{}{}
\crefname{lemma}{Lemma}{Lemmas}
\crefname{proposition}{Proposition}{Propositions}
\crefname{section}{Section}{Sections}
\crefname{table}{Table}{Tables}
\crefname{theorem}{Theorem}{Theorems}

\newcommand{\A}{\mathcal{A}}
\newcommand{\Adet}{\A_{\mathit{det}}}
\newcommand{\chk}{\mathit{chk}}
\newcommand{\B}{\mathcal{B}}
\newcommand{\Bdet}{\B_{\mathit{det}}}

\newcommand{\btran}[1]{\xhookrightarrow{#1}}
\newcommand{\coNBA}{\textup{coNBA}}
\newcommand{\coUBA}{\textup{coUBA}}
\newcommand{\Cyl}[1]{#1{\downarrow}}
\newcommand{\deltadet}{\delta_{\mathit{det}}}
\newcommand{\DPA}{\textup{DPA}}
\newcommand{\en}{\mathit{end}}
\newcommand{\Ex}{\mathbb{E}}

\newcommand{\finally}{\mathsf{F}}
\newcommand{\ftrees}[1]{\llparenthesis #1 \rrparenthesis}
\newcommand{\LTL}{\textup{LTL}}
\newcommand{\N}{\mathbb{N}}
\newcommand{\NBA}{\textup{NBA}}
\newcommand{\norm}[1]{\left\lVert#1\right\rVert}
\renewcommand{\P}{\mathbb{P}}
\newcommand{\Prob}{\mathit{Prob}}
\newcommand{\Q}{\mathbb{Q}}
\newcommand{\R}{\mathbb{R}}
\newcommand{\sacc}{s_{\mathit{acc}}}
\newcommand{\trees}[1]{\llbracket #1 \rrbracket}
\newcommand{\UBA}{\textup{UBA}}
\newcommand{\Xhookrightarrow}[2]{{\;\xhookrightarrow{#1}{}\negthickspace{#2}\;}}
\newcommand{\Xrightarrow}[2]{{\;\xrightarrow{#1}{}\negthickspace^{#2}\;}}

\newcommand{\BPtuple}{(\Gamma, \mathord{\btran{}}, \Prob, X_0)}

\title{Linear-Time Model Checking Branching Processes}

\author{Stefan Kiefer}{University of Oxford, UK}{}{}{supported by a Royal Society University Research Fellowship}
\author{Pavel Semukhin}{University of Oxford, UK}{}{}{}
\author{Cas Widdershoven}{University of Oxford, UK}{}{}{}
\authorrunning{S. Kiefer and P. Semukhin and C. Widdershoven}
\Copyright{Stefan Kiefer and Pavel Semukhin and Cas Widdershoven}



\category{} 

\relatedversion{This is the full version of a CONCUR 2021 paper.} 

\nolinenumbers 

\hideLIPIcs  

\EventEditors{John Q. Open and Joan R. Access}
\EventNoEds{2}
\EventLongTitle{42nd Conference on Very Important Topics (CVIT 2016)}
\EventShortTitle{CVIT 2016}
\EventAcronym{CVIT}
\EventYear{2016}
\EventDate{December 24--27, 2016}
\EventLocation{Little Whinging, United Kingdom}
\EventLogo{}
\SeriesVolume{42}
\ArticleNo{23}

\begin{document}

\maketitle

\begin{abstract}
(Multi-type) branching processes are a natural and well-studied model for generating random infinite trees.
Branching processes feature both nondeterministic and probabilistic branching, generalizing both transition systems and Markov chains (but not generally Markov decision processes).
We study the complexity of model checking branching processes against linear-time omega-regular specifications:
is it the case almost surely that every branch of a tree randomly generated by the branching process satisfies the omega-regular specification?
The main result is that for LTL specifications this problem is in PSPACE, subsuming classical results for transition systems and Markov chains, respectively.
The underlying general model-checking algorithm is based on the automata-theoretic approach, using unambiguous B\"uchi automata.
\end{abstract}

\section{Introduction} \label{sec:intro}

Checking whether a (labelled) transition system satisfies a linear-time specification is a staple in verification.
The specification is often given as a formula of linear temporal logic (LTL).
While early procedures for LTL model checking work directly with the formula \cite{LichtensteinP85}, the automata-theoretic approach translates LTL formulas into finite automata on infinite words, such as B\"uchi automata, and analyzes a product of the system and the automaton \cite{VardiWolper86}.
This approach can lead to clean and modular model-checking algorithms.

Although LTL captures only a subset of $\omega$-regular languages, model-checking algorithms based on the automata-theoretic approach can be made optimal from the point of view of computational complexity.
In particular, model checking finite transition systems against LTL specifications is PSPACE-complete~\cite{SistlaVW87}, and the algorithm \cite{VardiWolper86} that, loosely speaking, translates (the negation of) the LTL formula into a B\"uchi automaton and checks the product with the transition system for emptiness can indeed be implemented in PSPACE.

The same approach does not directly work for probabilistic systems modelled as finite Markov chains: intuitively, the nondeterminism in a B\"uchi automaton causes issues in a stochastic setting where the specification should hold with probability~$1$, i.e., almost surely but not necessarily surely.
A possible remedy is to translate the nondeterministic B\"uchi automaton further into a deterministic automaton, e.g., a deterministic Rabin automaton (deterministic B\"uchi automata are less expressive), with which the Markov chain can be naturally instrumented and subsequently analyzed.
This determinization step causes a (second) exponential blowup and does not lead to algorithms that are optimal from a computational-complexity point of view.
However, for \emph{Markov decision processes (MDPs)}, which allow for nondeterminism in the probabilistic system, this approach is adequate and leads to an optimal, double-exponential time, model-checking algorithm.

Checking whether a Markov chain satisfies an LTL specification with probability~$1$ is PSPACE-complete, but membership in PSPACE was proved only in \cite{CY88,CourcoubetisYannakakis95}, not using the automata-theoretic approach but by a recursive procedure on the formula.
This raised the question if there is also an optimal algorithm based on the automata-theoretic approach; see \cite{Vardi99} for a survey of the state of the art at the end of the 90s.

The answer is yes and was first given in \cite{CouSahSut03}, using a single-exponential translation from LTL to \emph{separated} B\"uchi automata.
Such automata are special \emph{unambiguous} B\"uchi automata, which restrict nondeterministic B\"uchi automata by requiring that every word have at most one accepting run.
Another algorithm, using alternating B\"uchi automata, was proposed in \cite{BustanRubinVardi04}, exploiting \emph{reverse determinism}, a property also related to unambiguousness.
A polynomial-time (even NC) model-checking algorithm for Markov chains against general unambiguous B\"uchi automata was given in \cite{16BKKKMW-CAV}.
These works all imply optimal PSPACE algorithms for LTL model checking of Markov chains via the automata-theoretic approach.

In this paper we exhibit an LTL model checking algorithm that has the following features: (1)~it applies to (multi-type) branching processes, a well established model for random trees, generalizing both nondeterministic transition systems and Markov chains; (2)~it runs in PSPACE, which is the optimal complexity both for nondeterministic transition systems and Markov chains; and (3)~it is based on the automata-theoretic approach (using unambiguous B\"uchi automata).
The fact that there exists an algorithm with the first two features might seem surprising, as one might think that any system model that encompasses both nondeterminism and probability will generalize MDPs, for which LTL model checking is 2EXPTIME-complete \cite{CourcoubetisYannakakis95}.

Branching processes (BPs) are a well-studied model in mathematics with applications in numerous fields including biology, physics and natural language processing; see, e.g., \cite{Harris63,AthreyaNey72,HaccouJV05}.
BPs randomly generate infinite trees, and, from a computer-science point of view, they might be the most natural model to do so:
(multi-type) BPs can be thought of as a version of stochastic context-free grammars without terminal symbols, randomly generating infinite derivation trees.
For example, consider the following BP, taken from \cite{12CDK-MFCS}, with $3$ \emph{types} $I, B, D$:
\begin{align}
 &I \btran{0.9} I     && B \btran{0.2} D && D \btran{1} D \nonumber \\[-1mm]
 &I \btran{0.1} I B   && B \btran{0.5} B \label{eq:intro-example} \\[-1mm]
 &                    && B \btran{0.3} B B \nonumber
\end{align}
This BP might generate a tree with the following prefix:
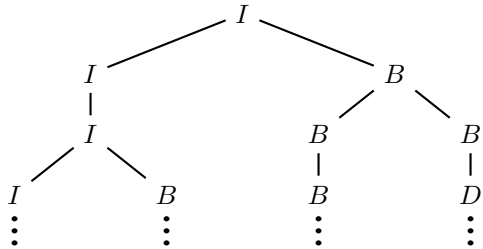
\begin{center}
\begin{tikzpicture}[yscale=0.8]
  \node (0) at (0,0) {$I$};
  \node (1) at (-2,-1) {$I$};
  \node (2) at (2,-1) {$B$};
  \node (11) at (-2,-2) {$I$};
  \node (111) at (-3,-3) {$I$};
  \node (112) at (-1,-3) {$B$};
  \node (21) at (1,-2) {$B$};
  \node (22) at (3,-2) {$B$};
  \node (211) at (1,-3) {$B$};
  \node (221) at (3,-3) {$D$};
  \draw (0)--(1);
  \draw (0)--(2);
  \draw (1)--(11);
  \draw (11)--(111);
  \draw (11)--(112);
  \draw (2)--(21);
  \draw (2)--(22);
  \draw (21)--(211);
  \draw (22)--(221);
  \draw[fill] (-3,-3.4) circle (0.7pt);
  \draw[fill] (-3,-3.6) circle (0.7pt);
  \draw[fill] (-3,-3.8) circle (0.7pt);
  \draw[fill] (-1,-3.4) circle (0.7pt);
  \draw[fill] (-1,-3.6) circle (0.7pt);
  \draw[fill] (-1,-3.8) circle (0.7pt);
  \draw[fill] (1,-3.4) circle (0.7pt);
  \draw[fill] (1,-3.6) circle (0.7pt);
  \draw[fill] (1,-3.8) circle (0.7pt);
  \draw[fill] (3,-3.4) circle (0.7pt);
  \draw[fill] (3,-3.6) circle (0.7pt);
  \draw[fill] (3,-3.8) circle (0.7pt);
\end{tikzpicture}
\end{center}
The probability that the BP generates a tree with the shown prefix is the product of the probabilities of the fired transition rules, i.e., (in breadth-first order) $0.1 \cdot 0.9 \cdot 0.3 \cdot 0.1 \cdot 0.5 \cdot 0.2$.

BPs generalize transition systems.
Consider the following transition system:
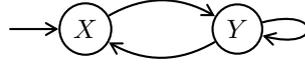
\begin{center}
\begin{tikzpicture}
\node[MDPrand] (X) at (0,0) {$X$};
\node[MDPrand] (Y) at (2,0) {$Y$};
\draw[->] (-1,0) to (X);
\draw[->] (X) edge[bend left] (Y);
\draw[->] (Y) edge[bend left] (X);
\draw[->] (Y) edge[loop right,looseness=12] 
(Y);
\end{tikzpicture}
\end{center}
It is equivalent to the BP with $X \btran{1} Y$ and $Y \btran{1} X Y$, which generates with probability~$1$ the following unique tree:
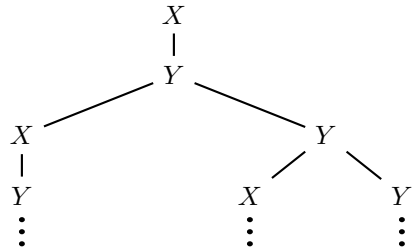
\begin{center}
\begin{tikzpicture}[yscale=0.8]
\node (0) at (0,0) {$X$};
\node (1) at (0,-1) {$Y$};
\node (2l) at (-2,-2) {$X$};
\node (2r) at ( 2,-2) {$Y$};
\node (3l) at (-2,-3) {$Y$};
\node (3rl) at (1,-3) {$X$};
\node (3rr) at (3,-3) {$Y$};
\draw[fill] (-2,-3.4) circle (0.7pt);
\draw[fill] (-2,-3.6) circle (0.7pt);
\draw[fill] (-2,-3.8) circle (0.7pt);
\draw[fill] ( 1,-3.4) circle (0.7pt);
\draw[fill] ( 1,-3.6) circle (0.7pt);
\draw[fill] ( 1,-3.8) circle (0.7pt);
\draw[fill] ( 3,-3.4) circle (0.7pt);
\draw[fill] ( 3,-3.6) circle (0.7pt);
\draw[fill] ( 3,-3.8) circle (0.7pt);
\draw (0) -- (1);
\draw (1) -- (2l);
\draw (1) -- (2r);
\draw (2l) -- (3l);
\draw (2r) -- (3rl);
\draw (2r) -- (3rr);
\end{tikzpicture}
\end{center}
The branches of this unique tree are exactly the executions of the transition system.
As a consequence, any LTL formula holds on all executions of the transition system if and only if it holds (with probability~$1$) on all branches of the generated tree.

BPs also generalize Markov chains.
Consider the following Markov chain:
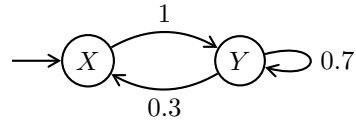
\begin{center}
\begin{tikzpicture}
\node[MDPrand] (X) at (0,0) {$X$};
\node[MDPrand] (Y) at (2,0) {$Y$};
\draw[->] (-1,0) to (X);
\draw[->] (X) edge[bend left] node[above] {$1$} (Y);
\draw[->] (Y) edge[bend left] node[below] {$0.3$} (X);
\draw[->] (Y) edge[loop right,looseness=12] node[right] {$0.7$} (Y);
\end{tikzpicture}
\end{center}
It is equivalent to the BP with $X \btran{1} Y$ and $Y \btran{0.3} X$ and $Y \btran{0.7} Y$, which generates, with probabilities $0.3$, $0.7 \cdot 0.3$, $0.7 \cdot 0.7$, respectively, the following prefixes of (degenerated) trees:
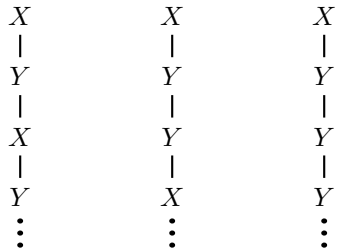
\begin{center}
\begin{tikzpicture}[yscale=0.8]
\node (00) at (0, 0) {$X$};  \node(10) at (2, 0) {$X$};  \node(20) at (4, 0) {$X$};
\node (01) at (0,-1) {$Y$};  \node(11) at (2,-1) {$Y$};  \node(21) at (4,-1) {$Y$};
\node (02) at (0,-2) {$X$};  \node(12) at (2,-2) {$Y$};  \node(22) at (4,-2) {$Y$};
\node (03) at (0,-3) {$Y$};  \node(13) at (2,-3) {$X$};  \node(23) at (4,-3) {$Y$};
\draw[fill] (0,-3.4) circle (0.7pt);
\draw[fill] (0,-3.6) circle (0.7pt);
\draw[fill] (0,-3.8) circle (0.7pt);
\draw[fill] (2,-3.4) circle (0.7pt);
\draw[fill] (2,-3.6) circle (0.7pt);
\draw[fill] (2,-3.8) circle (0.7pt);
\draw[fill] (4,-3.4) circle (0.7pt);
\draw[fill] (4,-3.6) circle (0.7pt);
\draw[fill] (4,-3.8) circle(0.7pt);
\draw (00) -- (01) -- (02) -- (03);
\draw (10) -- (11) -- (12) -- (13);
\draw (20) -- (21) -- (22) -- (23);
\end{tikzpicture}
\end{center}
Here, each possible ``tree'' has only a single branch, and the possible ``trees'' are distributed in the same way as the possible executions of the Markov chain.
As a consequence, any LTL formula holds with probability~$1$ on a random execution of the Markov chain if and only if it holds with probability~$1$ on the (single) branch of the generated tree.

Hence, both for the transition system and for the Markov chain, the respective model-checking question reduces to the \emph{BP model-checking problem} which asks whether with probability~$1$ the property holds on \emph{all} branches.

For LTL specifications, we refer to this BP model-checking problem as $\P(\LTL)=1$.
Our main result is that it is in PSPACE, generalizing the corresponding classical results on transition systems and Markov chains.
As mentioned, our model-checking algorithm is based on the automata-theoretic approach, in particular on unambiguous B\"uchi automata.
Another important technical ingredient is the algorithmic analysis of certain nonnegative matrices in terms of their spectral radius.

The latter points to the fact that the numbers in the system generally matter, even though we only consider the qualitative problem of comparing the satisfaction probability with~$1$.
For example, for the BP given in \eqref{eq:intro-example}, one can show that the probability that all branches eventually hit a node of type~$D$ is less than~$1$ (in fact, it is~$0$).
Intuitively, this is because the probability of ``branching'' via $B \btran{0.3} B B$ is larger than the probability of ``dying'' via $B \btran{0.2} D$.
Were the probabilities $0.3$ and $0.2$ swapped, the probability that all branches eventually hit a node of type~$D$ would be~$1$; cf.~\cite[Section~1]{12CDK-MFCS}.

We also consider the problem $\P(\LTL = 0)$, which asks whether the probability that \emph{all} branches satisfy a given LTL formula is~$0$.
Even though it is trivial to negate an LTL formula, this problem is (unlike in Markov chains) not equivalent to the complement of $\P(\LTL = 1)$, because even when the probability is less than~$1$ that the formula holds on all branches, the probability may still be~$0$ that the negated formula holds on all branches.
We will show that $\P(\LTL = 0)$ is much more computationally complex than $\P(\LTL=1)$: it is 2EXPTIME-complete.

Besides LTL, we also consider automata-based specifications.
B\"uchi automata are relevant from a verification point of view, as a way of specifying desired or undesired executions of the system.
Unambiguous B\"uchi automata are useful from a technical point of view, in particular, to facilitate our main result on $\P(\LTL=1)$.
See \cref{sec:prelims,tab:map} for definitions of our problems and a map of our results.

\begin{remark}
Readers familiar with MDPs may wonder how the problem $\P(\LTL)=1$ can have lower computational complexity than the problem whether all schedulers of an MDP satisfy an LTL specification almost surely.
Consider the BP
\begin{alignat*}{6}
X &\btran{1} Y_1 Y_2  & \qquad Y_1 &\btran{0.7} X & \qquad Y_1 &\btran{0.3} Z & \qquad Y_2 &\btran{0.5} X & \qquad Y_2 &\btran{0.5} Z& \qquad Z &\btran{1} Z\,,
\end{alignat*}
which might be depicted graphically as follows:
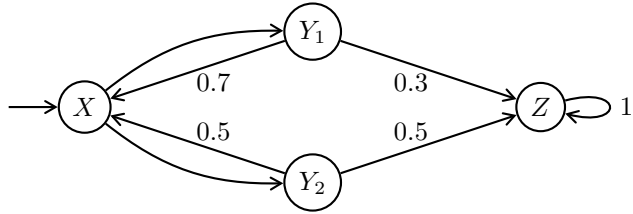
\begin{center}
\begin{tikzpicture}
\node[MDPrand] (X) at (0,0) {$X$};
\node[MDPrand] (Y1) at (3,1) {$Y_1$};
\node[MDPrand] (Y2) at (3,-1) {$Y_2$};
\node[MDPrand] (Z) at (6,0) {$Z$};
\draw[->] (-1,0) to (X);
\draw[->] (X) edge[bend left=20] (Y1);
\draw[->] (X) edge[bend right=20] (Y2);
\draw[->] (Y1) edge[bend left=0] node[pos=0.4,below] {$0.7$} (X);
\draw[->] (Y2) edge[bend right=0] node[pos=0.4,above] {$0.5$} (X);
\draw[->] (Y1) edge node[pos=0.4,below] {$0.3$} (Z);
\draw[->] (Y2) edge node[pos=0.4,above] {$0.5$} (Z);
\draw[->] (Z) edge[loop right,looseness=12] node[right] {$1$} (Z);
\end{tikzpicture}
\end{center}
One might view this BP as an MDP where in an $X$-node the scheduler nondeterministically picks either the $Y_1$- or the $Y_2$-successor, and in an $Y_i$-node, the $X$- or the $Z$-successor is chosen randomly.
In such an MDP, regardless of the scheduler, a random run reaches with probability~$1$ a $Z$-node.
However, in the BP above, the probability is positive that \emph{some} branch of a random tree never reaches a $Z$-node.
Although each branch of a random tree could be thought of as being witnessed by at least one scheduler, this is not a contradiction, as there are uncountably many schedulers (over which one cannot take a sum).
Hence, if an MDP is interpreted as a BP in the way sketched above, then the requirement that the BP satisfy an LTL formula almost surely on \emph{all} branches is \emph{stronger}, and computationally less complex to check, than the requirement that the MDP satisfy, for \emph{each} scheduler, the formula almost surely.
\end{remark}

\subparagraph*{Related work.}
We have already discussed related work concerning model checking transition systems and Markov chains.

In addition to the mentioned applications of BPs in various fields, there has also been work on BPs in computer science, especially in the last 10~years.
This paper builds on \cite{12CDK-MFCS}, where specifications in terms of deterministic parity \emph{tree} automata are considered.
The work \cite{12CDK-MFCS} implies decidability of the problems considered in this paper and some basic upper complexity bounds.
For example, it is not hard to derive from \cite{12CDK-MFCS} that $\P(\LTL=1)$ is in 2EXPTIME.
Lowering this to PSPACE is the main achievement of this paper.

A related strand of work considers \emph{regular tree languages}; i.e., the specification is not in terms of a word automaton that is run on each branch but in terms of tree automata.
Even measurability is not easy to show in this case \cite{GogaczMMS17}, and fundamental decidability questions around computing the measure have been answered positively only for subclasses of regular tree languages \cite{MichalewskiM15,NiwinskiPS20}.

Fundamental results on the complexity of algorithmically analyzing BPs have been obtained in \cite{EtessamiYannakakis09}.
Indeed, in \cref{sec:as-finite} we build on and improve results from \cite{EtessamiYannakakis09} on \emph{finiteness} (more often called ``extinction'' in the literature) of BPs.

Another recent line of work considers extensions of BPs with nondeterminism, focusing on algorithmic questions about properties such as reachability.
\emph{Branching MDPs}, which are BPs where a controller chooses actions to influence the evolution of the tree, have been investigated, e.g., in \cite{EtessamiSY18,EtessamiSY20}.
Even branching \emph{games}, featuring two adversarial controllers, have been studied recently \cite{EtessamiMSY19}.

The work \cite{GorlinR18} also considers BPs with ``internal'' nondeterminism (as opposed to the ``external'' nondeterminism manifested as branching in the generated tree), along with model-checking problems against the logic \emph{GPL}.
This expressive, $\mu$-calculus based modal logic had been introduced in \cite{CleavelandIN05}.
The system model therein, called \emph{reactive probabilistic labeled transition systems (RPLTSs)}, is essentially equivalent to BPs as considered in this paper.

BPs are related to models for probabilistic programs with recursion, such as \emph{Recursive Markov chains}, for which model-checking problems have been studied in detail; see, in particular, \cite{EtessamiYannakakis12}.
Very loosely speaking, a run of a (``1-exit'') Recursive Markov chain can be viewed as a depth-first traversal of a tree generated by a BP.
Indeed, for a lower bound in the present paper (\cref{thm:NBA-0}) we adapt a proof from \cite{EtessamiYannakakis12}.
However, most qualitative model-checking problems for Recursive Markov chains are EXPTIME-complete~\cite{EtessamiYannakakis12}, and so many of the BP problems we study turn out to have different computational complexity.

As a key technical tool we use unambiguous B\"uchi automata, as recently proposed for Markov chains~\cite{16BKKKMW-CAV}.
It is non-trivial to extend their use to random trees, as the branching behaviour of BPs interferes with the spectral-radius based analysis from~\cite{16BKKKMW-CAV}.
One may view as the main technical insight of this paper that the limited nondeterminism in unambiguous automata can be combined with the tree branching of BPs, so that, in a sense, BP model checking reduces to comparing the spectral radius of a certain nonnegative matrix with~$1$ (\cref{prop:coUBA-1}).

\section{Preliminaries} \label{sec:prelims}

Let $\N$ and $\N_0$ denote the set of positive and nonnegative integers, respectively.
For a finite set $\Gamma$, we write $\Gamma^*$ (resp., $\Gamma^+$) for the set of words (resp., nonempty words) over~$\Gamma$.

\subparagraph*{Branching processes.}
A \emph{(multi-type) branching process (BP)} is a tuple $\B  =  \BPtuple$, where $\Gamma$ is a finite set of types,
 $\mathord{\btran{}} \subseteq \Gamma\times \Gamma^+$ is a finite set of transition rules,
 $\Prob$ is a function assigning positive rational probabilities to transition rules
  so that for every $X\in \Gamma$ we have $\sum_{X \btran{} w} \Prob(X \btran{} w)=1$,
 and $X_0\in \Gamma$ is the start type.
We write $X \btran{p} w$ to denote that $\Prob(X \btran{} w) = p$.
Given a BP~$\B$ and a type $X \in \Gamma$ we write $\B[X]$ for the BP obtained from~$\B$ by making $X$ the start type.
For $X, Y \in \Gamma$ we call $Y$ a \emph{successor} of~$X$ if there is a rule $X \btran{} u Y v$ for some $u, v \in \Gamma^*$.

A BP \emph{with $\varepsilon$-rules allowed} relaxes the requirement $\mathord{\btran{}} \subseteq \Gamma\times \Gamma^+$ to $\mathord{\btran{}} \subseteq \Gamma\times \Gamma^*$, i.e., there may be rules of the form $X \btran{} \varepsilon$, where $\varepsilon$ denotes the empty word.
In the following, we disallow $\varepsilon$-rules unless specified otherwise; but the definitions generalize in a natural way.

Fix a BP $\B = \BPtuple$ for the rest of the section.

\subparagraph*{Trees.}
Write $\trees{\B}$ for the set of trees \emph{generated} by~$\B$; i.e., $\trees{\B}$ denotes the set of ordered $\Gamma$-labelled trees~$t$ such that for each $X \in \Gamma$ and each $X$-labelled node~$v$ in~$t$, there is a rule $X \btran{} X_1 \cdots X_k$, denoted by $\mathit{rule}(v)$, such that the $k$ ordered children of~$v$ are labelled with $X_1, \ldots, X_k$, respectively.
We say a node \emph{has type $X \in \Gamma$} if the node is labelled with~$X$.
A \emph{finite prefix} of a tree $t \in \trees{\B}$ is an ordered $\Gamma$-labelled \emph{finite} tree obtained from~$t$ by designating some nodes as leaves, and removing all their children, grandchildren, etc.
Write $\ftrees{\B}$ for the set of finite prefixes of trees generated by~$\B$.
For $t \in \ftrees{\B}$ write $\Cyl{t} \subseteq \trees{\B}$ for the (``cylinder'') set of trees $t' \in \trees{\B}$ such that $t$ is a finite prefix of~$t'$.
For $X \in \Gamma$ write $\trees{\B}_X \subseteq \trees{\B}$ and $\ftrees{\B}_X \subseteq \ftrees{\B}$ for the subsets of trees whose root has type~$X$; the trees in $\trees{\B}_X$ are called \emph{$X$-trees}.
A \emph{branch} of a tree~$t$ is a sequence $v_0 v_1 \cdots$ of nodes in~$t$, where $v_0$ is the root of~$t$ and $v_{i+1}$ is a child of~$v_i$ for all $i \in \N_0$.
See~\cite{12CDK-MFCS} for equivalent, more formal tree-related definitions.

\subparagraph*{Probability space.}
For each $X \in \Gamma$ we define the probability space $(\trees{\B}_X,\Sigma_X,\P_X)$, where $\Sigma_X$ is the $\sigma$-algebra generated by $\{\Cyl{t} \mid t\in\ftrees{\B}_X\}$, and $\P_X$ is the probability measure generated by $\P_X(\Cyl{t})  :=  \prod_{v} \Prob(\mathit{rule}(v))$
for all $t \in \ftrees{\B}_X$, where the product extends over all non-leaf nodes~$v$ in~$t$.
This is analogous to the standard definition of the probability space of a Markov chain.
We may write $\P_\B$ for $\P_{X_0}$, omitting the subscript when $\B$ is understood.
We often talk about events (i.e., measurable sets of trees) and their probability in text form.
For example, by saying ``a $\B$-tree has with positive probability infinitely many nodes of type~$X$'' we mean that $\P_\B(E) > 0$ where $E \subseteq \trees{\B}_{X_0}$ is the set of $X_0$-trees with infinitely many nodes of type~$X$.

\subparagraph*{Linear-Time Properties.}
We are particularly interested in sets of trees \emph{all whose branches} (more precisely, their associated sequences of types) satisfy an $\omega$-regular linear-time property $L \subseteq \Gamma^\omega$.
Given $L \subseteq \Gamma^\omega$, we write $\P_\B(L)$ for the probability that \emph{all} branches of a $\B$-tree satisfy~$L$.
\emph{Linear temporal logic (LTL)} formulas specify linear-time properties; see, e.g., \cite{Thomas90} for a definition of LTL.
An important example for us are formulas of the form $\finally T$, where $T \subseteq \Gamma$, which denotes the linear-time property $\{u X w \mid u \in \Gamma^*,\ X \in T,\ w \in \Gamma^\omega\}$.
Accordingly, $\P_\B(\finally T)$ denotes the probability that all branches of a $\B$-tree have a node whose type is in~$T$ (equivalently, the probability that a $\B$-tree has a finite prefix all whose leaves have a type in~$T$).

\subparagraph*{Automata.}
We use finite automata on infinite words over~$\Gamma$, where $\Gamma$ is the set of types of a BP.
We use \emph{deterministic parity automata (DPAs)}, \emph{deterministic B\"uchi automata (DBAs)}, \emph{nondeterministic B\"uchi automata (NBAs)}, and \emph{unambiguous B\"uchi automata (UBAs)}.
The definitions are standard; see, e.g., \cite{Thomas90}.
In the following we fix some terms and notation.
Let $\A = (Q, \Gamma, \delta, Q_0, F)$ be an NBA, where $Q$ is a finite set of states, $\Gamma$ is the alphabet, $\delta \subseteq Q \times \Gamma \times Q$ is the transition relation, $Q_0 \subseteq Q$ is the set of initial states, and $F \subseteq Q$ is the set of accepting states.
We write $q \xrightarrow{X} r$ to denote that $(q,X,r) \in \delta$.
A finite sequence $q_0 \xrightarrow{X_1} q_1 \xrightarrow{X_2} \cdots \xrightarrow{X_n} q_n$ is called a \emph{path} and can be summarized as $q_0 \Xrightarrow{X_1 \cdots X_n}* q_n$. 
An infinite sequence $q_0 \xrightarrow{X_1} q_1 \xrightarrow{X_2} \cdots$ is called a \emph{run} of $X_1 X_2 \cdots$.
We call the run \emph{accepting} if $q_0 \in Q_0$ and $q_i \in F$ holds for infinitely many~$q_i$.
The NBA~$\A$ \emph{accepts} (resp., \emph{rejects}) an infinite word $w \in \Gamma^\omega$ if $w$ has (resp., does not have) an accepting run in~$\A$.
The NBA~$\A$ is called an \emph{unambiguous B\"uchi automaton (UBA)} if every $w \in \Gamma^\omega$ has at most one accepting run.
An automaton~$\A$ defines $\omega$-regular linear-time properties $\{w \in \Gamma^\omega \mid \A \text{ accepts } w\}$ and $\{w \in \Gamma^\omega \mid \A \text{ rejects } w\}$.
In keeping with previous definitions, we write $\P_\B(\A \text{ accepts})$ (resp., $\P_\B(\A \text{ rejects})$) for the probability that \emph{all} branches of a $\B$-tree (more precisely, their associated sequences of types) are accepted (resp., rejected) by~$\A$.

\subparagraph*{Problems.}
We consider the following computational problems.
The problem $\P(\text{finite}) = 1$ asks, given a BP~$\B$ with $\varepsilon$-rules allowed, whether the probability that a $\B$-tree is finite is~$1$.
The problem $\P(\LTL) = 1$ asks, given a BP~$\B$ and an LTL formula~$\varphi$, whether $\P_\B(\varphi) = 1$.
The problems $\P(\DPA) = 1$ (resp., $\P(\NBA) = 1$) ask, given a BP~$\B$ and a DPA (resp., NBA)~$\A$, whether $\P_\B(\A \text{ accepts}) = 1$.
The problems $\P(\coNBA) = 1$ (resp., $\P(\coUBA) = 1$)\footnote{We do not explicitly define or use a notion of ``co-B\"uchi automata'' to avoid possible confusion about accepting/rejecting. If one were to do so, one would define a ``co-NBA''~$\A$ like an NBA~$\A$, but the ``co-NBA''~$\A$ would accept a word $w \in \Gamma^\omega$ if and only if $\A$ viewed as an NBA rejects~$w$. Similarly for ``co-UBAs''.} ask, given a BP~$\B$ and an NBA (resp., UBA)~$\A$, whether $\P_\B(\A \text{ rejects}) = 1$.
The problems $\P(\LTL) = 0,  \P(\DPA) = 0, \ldots$ are defined similarly, where ``${=}\,1$'' is replaced with ``${=}\,0$''.
\begin{table}
\begin{center}
\begin{tabular}{l|cc}
                     & $=1$                      & $=0$               \\ \hline
$\P(\text{finite})$  & in NC                                          \\
\cref{sec:as-finite} & \cref{prop:as-finite-NC}  &                    \\[2mm]
$\P(\DPA)$           & in NC                     & P                  \\
\cref{sec:DPA}       & \cref{thm:DPA-1}          & \cref{thm:DPA-0}   \\[2mm]
$\P(\NBA)$           & PSPACE                    & EXPTIME            \\
\cref{sec:NBA}       & \cref{thm:NBA-1}          & \cref{thm:NBA-0}
\end{tabular}
\hfill
\begin{tabular}{l|cc}
                     & $=1$                      & $=0$               \\ \hline
$\P(\coNBA)$         & PSPACE                    & EXPTIME            \\
\cref{sec:coNBA}     & \cref{thm:coNBA-1}        & \cref{thm:coNBA-0} \\[2mm]
$\P(\coUBA)$         & in NC                     &                    \\
\cref{sec:coUBA}     & \cref{prop:coUBA-1}       &                    \\[2mm]
$\P(\LTL)$           & PSPACE                    & 2EXPTIME           \\
\cref{sec:LTL}       & \cref{thm:LTL-1}          & \cref{thm:LTL-0}
\end{tabular}
\end{center}
\caption{Results and organization of the paper.
The complexity classes indicate completeness results, except ``in NC'', which only means membership in NC.}
\label{tab:map}
\end{table}
See \cref{tab:map} for a map of our results in those terms, as well as for an overview of the rest of the paper.
As explained in the introduction, the problem $\P(\LTL)=1$ 
is of particular interest from a model-checking point of view, and the technically most challenging one.

\subparagraph*{Complexity Classes.}
In addition to standard complexity classes between P and 2EXPTIME, we use the class NC, the subclass of P comprising those problems solvable in polylogarithmic time by a parallel random-access machine using polynomially many processors; see, e.g., \cite[Chapter 15]{Pap94}.
To prove membership in PSPACE in a modular way, we will use the following pattern:
\begin{lemma} \label{lem:PSPACE-transducer}
Let $P_1, P_2$ be two problems, where $P_2$ is in NC.
Suppose there is a reduction from $P_1$ to~$P_2$ implemented by a PSPACE transducer, i.e., a Turing machine whose work tape (but not necessarily its output tape) is PSPACE-bounded.
Then $P_1$ is in PSPACE.
\end{lemma}
\begin{proof}
Note that the output of the transducer is (at most) exponential.
Problems in NC can be decided in polylogarithmic space~\cite[Theorem~4]{Borodin77}.
Using standard techniques for composing space-bounded transducers (see, e.g., \cite[Proposition~8.2]{Pap94}), it follows that $P_1$ is in PSPACE.
\end{proof}

\subparagraph*{Matrices.}
We use finite sets $S$ to index matrices $M \in \R^{S \times S}$ and vectors $v \in \R^S$.
The \emph{graph} of a nonnegative matrix $M \in [0, \infty)^{S \times S}$ is the directed graph $(S,E)$ with $E = \{(s,t) \in S \times S \mid M_{s,t} > 0\}$.
The \emph{spectral radius} of a matrix is the largest absolute value of its eigenvalues.
The following lemma allows to efficiently compare the spectral radius of a nonnegative matrix with~$1$.
\begin{lemma} \label{lem:determine-spectral-radius}
Given a nonnegative rational matrix~$M$, one can determine in NC whether $\rho < 1$ or $\rho = 1$ or $\rho > 1$, where $\rho$ denotes the spectral radius of~$M$.
\end{lemma}
\begin{proof}
Use the algorithm from \cite[Proposition~2.2]{13EGK-IPL}, but not with Gaussian elimination as suggested there, but by solving the systems of linear equations described in \cite[Proposition~2.2]{13EGK-IPL} in~NC.
The latter is possible in NC~\cite[Theorem~5]{BorodinGathenHopcroft82}. 
\end{proof}

\section{Basic Results} \label{sec:basic-results}

In this section we develop the more basic results indicated in \cref{tab:map}, on finiteness (\cref{sec:as-finite}), deterministic parity automata (\cref{sec:DPA}), and B\"uchi automata (\cref{sec:NBA}), on the one hand rounding off the complexity map in \cref{tab:map}, and on the other hand building the foundation for more challenging results in the following sections.
In particular, \cref{prop:as-finite-NC} is indirectly used throughout the paper.

\subsection{Finiteness} \label{sec:as-finite}

In this section we consider BPs with $\varepsilon$-rules allowed, i.e., rules of the form $X \btran{} \varepsilon$.
Such BPs may generate finite trees.
We are interested in the \emph{almost-sure finiteness} problem, also denoted as $\P(\text{finite})=1$, i.e., the problem whether the probability that a given BP with $\varepsilon$-rules allowed generates a finite tree is equal to~$1$.
In \cref{prop:as-finite-NC} below we show that this problem is in~NC.
All upper bounds on the complexity of $\P(\cdot) = 1$ problems in this paper build directly or indirectly on this result.

While the almost-sure finiteness (or ``extinction'') problem has often been studied and is known to be in (strongly) polynomial time \cite{EtessamiYannakakis09,13EGK-IPL}, its membership in NC is, to the best of the authors' knowledge, new.
For instance, since linear programming is P-complete, one cannot use linear programming (as in \cite{EtessamiYannakakis09}) to show membership in~NC.
Nor can one directly use the strongly polynomial-time algorithm of~\cite{13EGK-IPL}, as it computes, in a sub-procedure, the set of types~$X$ for which there \emph{exists} a finite $X$-tree.
But the latter problem is P-complete.

For the rest of the section, fix a BP $\B = \BPtuple$ with $\varepsilon$-rules allowed.
Define a directed graph $G = (\Gamma,E)$ (i.e., the types of~$\B$ are the vertices of~$G$) with an edge $(X,Y) \in E$ if and only if $Y$ is a successor of~$X$ (i.e., there is a rule $X \btran{} u Y v$ for some $u,v \in \Gamma^*$).
Given a strongly connected component (SCC) $S \subseteq \Gamma$ of~$G$ and $X \in S$, define a BP $\B[S,X] = (S,\mathord{\btran{}_S},\Prob_S,X)$ obtained from~$\B$ by restricting the types to~$S$ and deleting on all right-hand sides of the rules those types not in~$S$.
The following lemma is straightforward:%
\begin{lemma} \label{lem:as-finite-SCC}
A $\B$-tree is infinite with positive probability if and only if there exist an SCC $S \subseteq \Gamma$ of~$G$ and $X \in S$ such that $X$ is reachable from~$X_0$ in~$G$ and a $\B[S,X]$-tree is infinite with positive probability.
\end{lemma}
Let $M \in \Q^{\Gamma \times \Gamma}$ be the nonnegative $\Gamma \times \Gamma$-matrix with $M_{X,Y} = \sum_{X \btran{p} w} p |w|_Y$, where $|w|_Y \in \N_0$ is the number of occurrences of~$Y$ in~$w$.
That is, $M_{X,Y}$ is the expected number of direct $Y$-successors of the root of a $\B[X]$-tree.
By induction, $M^i$, the $i$th power of~$M$, is such that $(M^i)_{X,Y}$ is the expected number of $Y$-nodes that are exactly $i$ levels under the root of a $\B[X]$-tree.
The graph of~$M$ is exactly the previously defined graph~$G$.

Let $S \subseteq \Gamma$ be an SCC of~$G$.
Denote by $M_S \in \Q^{S \times S}$ the (square) principal submatrix obtained from~$M$ by restricting it to the rows and columns indexed by elements of~$S$.
Let $\rho_S$ denote the spectral radius of~$M_S$.
Call $S$ \emph{supercritical} if $\rho_S > 1$.
Call $S$ \emph{linear} if for all rules $X \btran{} w$ with $X \in S$ there is exactly one occurrence in~$w$ of a type in~$S$.
Observe that if $S$ is linear then $M_S$ is \emph{stochastic}, i.e., $M_S \vec{1} = \vec{1}$ where $\vec{1}$ is the all-$1$ vector, i.e., the element of $\{1\}^S$.
In that case, by the Perron-Frobenius theorem \cite[Theorem~2.1.4~(b)]{BermanP94}, we have $\rho_S = 1$ and, thus, $S$ is not supercritical.

The following characterization can be proved using \cite[Section~3]{13EGK-IPL} (which builds on \cite[Section~8.1]{EtessamiYannakakis09}):
\begin{restatable}{lemma}{lemasfinitenesschar} \label{lem:as-finiteness-char}
A $\B$-tree is infinite with positive probability if and only if there exist an SCC $S \subseteq \Gamma$ of~$G$ and $X \in S$ such that $X$ is reachable from~$X_0$ in~$G$ and $S$ is supercritical or linear.
\end{restatable}

It follows:

\begin{restatable}{proposition}{propasfiniteNC} \label{prop:as-finite-NC}
The problem $\P(\text{finite})=1$ is in NC.
\end{restatable}

\subsection{Deterministic Parity Automata} \label{sec:DPA}

In this section we consider deterministic parity automata (DPAs) on words.
In \cite[Section~3]{12CDK-MFCS} it was shown that the problem $\P(\DPA) = 1$ can be decided in polynomial time.
We improve this to membership in~NC.

By the following lemma we can check in NC whether a $\B$-tree almost surely has a finite prefix all whose leaves have types in a given set~$T$.
The proof is by reduction to almost-sure finiteness.
\begin{restatable}{lemma}{lemAFTone} \label{lem:AFT-1}
Given a BP $\B = \BPtuple$ and a set of types $T \subseteq \Gamma$, the problem whether $\P_{X_0}(\finally T) = 1$ is in NC.
\end{restatable}

By combining \cref{lem:AFT-1} with results from~\cite{12CDK-MFCS} we obtain:
\begin{restatable}{theorem}{thmDPAone} \label{thm:DPA-1}
The problem $\P(\DPA) = 1$ is in NC.
\end{restatable}

The hardness result in the following theorem highlights the different complexities of $\P(\cdot) = 0$ and $\P(\cdot) = 1$ problems in this paper.

\begin{restatable}{theorem}{thmDPAzero} \label{thm:DPA-0}
The problem $\P(\DPA) = 0$ is P-complete.
It is P-hard even for deterministic B\"uchi automata with two states, the accepting state being a sink.
\end{restatable}

\subsection{B\"uchi Automata} \label{sec:NBA}

\begin{theorem} \label{thm:NBA-1}
The problem $\P(\NBA) = 1$ is PSPACE-complete.
\end{theorem}
\begin{proof}
PSPACE-hardness is immediate in two different ways.
It follows from the PSPACE-hardness of model checking Markov chains against NBAs~\cite{Vardi85}.
It also follows from the PSPACE-hardness of model checking transition systems against NBAs. 
(The latter follows easily from the PSPACE-hardness of NBA universality~\cite{SistlaVW87}.)
Both model-checking problems are special cases of $\P(\NBA) = 1$.

Towards membership in PSPACE, we use a translation from NBA to DPA~\cite{Piterman07}.
This translation causes an exponential blow-up, but an inspection of the construction~\cite[Section~3.2]{Piterman07} reveals that it can be computed by a PSPACE transducer.
By \cref{thm:DPA-1} the problem $\P(\DPA) = 1$ is in NC.
By \cref{lem:PSPACE-transducer} it follows that $\P(\NBA) = 1$ is in PSPACE.
\end{proof}

\begin{restatable}{theorem}{thmNBAzero} \label{thm:NBA-0}
The problem $\P(\NBA) = 0$ is EXPTIME-complete.
It is EXPTIME-hard even for NBAs whose only accepting state is a sink.
\end{restatable}
\begin{proof}
Towards membership in EXPTIME, an NBA can be translated, in exponential time, to a DPA of exponential size; see, e.g., \cite{Piterman07}.
Since $\P(\DPA) = 0$ is in P by \cref{thm:DPA-0}, it follows that $\P(\NBA) = 0$ is in EXPTIME.

Concerning EXPTIME-hardness, we adapt the proof (in the online appendix) of \cite[Theorem 17]{EtessamiYannakakis12} on model checking \emph{recursive Markov chains} against NBAs.
The details are in \cref{app:NBA-0}.
\end{proof}

\section{Co-B\"uchi Automata} \label{sec:coNBA}

In this section we consider the problem $\P(\coNBA) = 1$, which asks, given a BP~$\B$ and a B\"uchi automaton~$\A$, whether $\B$ almost surely generates a tree whose branches are all rejected by~$\A$; i.e., whether $\P_{\B}(\A \text{ rejects}) = 1$.
Dually, one might ask whether the probability is positive that a $\B$-tree has a branch accepted by~$\A$.
Intuitively, we view the B\"uchi automaton~$\A$ as specifying ``bad'' branches, and we would like the tree almost surely not to have any bad branches.

This problem is in PSPACE, which can be shown via a translation to DPAs, as in \cref{thm:NBA-1}.
However, with a view on the following sections, in particular on LTL specifications, we pursue a different approach to the problem $\P(\coNBA) = 1$.
In this section we lay the groundwork for arbitrary B\"uchi automata~$\A$.
By building on these results, we will show in the next section that if $\A$ is \emph{unambiguous} then the problem is in~NC, which will allow us to derive our headline result, namely that $\P(\LTL)=1$ is in PSPACE.

Let $\B = \BPtuple$ be a BP and $\A = (Q, \Gamma, \delta, Q_0, F)$ a (not necessarily unambiguous) B\"uchi automaton.

Define a B\"uchi automaton, $\A \times \B$, by $\A \times \B  := (Q \times \Gamma, \Gamma, \delta_{\A \times \B}, Q_0 \times \{X_0\}, F \times \Gamma)$, where
\begin{align*}
 &\delta_{\A \times \B}((q_1, X_1), X_2)
 =\begin{cases}
 \delta(q_1, X_1) \times \{X_2\} & \text{if $X_2$ is a successor of~$X_1$}\\
 \emptyset                       & \text{otherwise.}
\end{cases}
\end{align*}

The remainder of the section is organized as follows.
In \cref{sub:UBA-X1-f} we show that the problem $\P(\coNBA) = 1$ reduces to the analysis of certain SCCs within $\A \times \B$.
In \cref{sub:Bdet} we introduce a key lemma, \cref{lem:Bdet}, which allows us to ``forget'' about the distinction between accepting and non-accepting states:
the lemma reduces $\P(\coNBA) = 1$ to a pure reachability problem in an exponential-sized BP, $\Bdet$.
This leads us to prove PSPACE-completeness of $\P(\coNBA) = 1$, but more importantly, \cref{lem:Bdet} plays a key role in the rest of the paper.
We prove it in \cref{sub:Bdet-proof}.

\subsection{The Automaton \texorpdfstring{$\A[f,X_f]$}{A[f,Xf]}} \label{sub:UBA-X1-f}

For any $(f,X_f) \in F \times \Gamma$ on a cycle of the transition graph of $\A \times \B$, define the B\"uchi automaton \[\A[f,X_f] \ := \ (\{\bar{q}_0\} \cup Q[f,X_f], \Gamma, \delta[f,X_f], \{\bar{q}_0\}, \{(f,X_f)\})\] as the B\"uchi automaton obtained from $\A \times \B$ by
\begin{enumerate}
\item making $(f,X_f)$ the only accepting state,
\item restricting the set of states, $Q[f,X_f] \subseteq Q \times \Gamma$, to those $(q,X)$ that, in the transition graph of~$\A \times \B$, are reachable from~$(f,X_f)$ and can reach~$(f,X_f)$, i.e., those $(q,X)$ in the SCC containing $(f, X_f)$,
\item restricting the transition function $\delta[f,X_f]$ accordingly, i.e., \[\delta[f,X_f]((q,X),Y) \ := \ \delta_{\A \times \B}((q,X),Y) \cap Q[f,X_f]\,,\]
\item making $\bar{q}_0$ the only initial state, and
\item setting $\delta[f,X_f](\bar{q}_0,X_f) := \{(f,X_f)\}$ and $\delta[f,X_f](\bar{q}_0,X) := \emptyset$ for all $X \in \Gamma \setminus \{X_f\}$.
\end{enumerate}

The following lemma follows from the pigeonhole principle and basic probability arguments:%
\begin{restatable}{lemma}{lemUBAXonef} \label{lem:UBA-X1-f}
The probability that some branch of a $\B$-tree is accepted by~$\A$ is positive if and only if
there are $q_0 \in Q_0$ and $f \in F$ and $X_f \in \Gamma$ such that $(f, X_f)$ is reachable from~$(q_0, X_0)$ in the transition graph of $\A \times \B$ and the probability that some branch of a $\B[X_f]$-tree is accepted by~$\A[f,X_f]$ is positive.
\end{restatable}

For the rest of the section let $(f,X_f) \in F \times \Gamma$ be on a cycle of the transition graph of $\A \times \B$.

\subsection{The Determinization \texorpdfstring{$\Adet$}{Adet} and the BP \texorpdfstring{$\Bdet$}{Bdet}} \label{sub:Bdet}

Let \[\Adet := (2^{\{\bar{q}_0\} \cup Q[f,X_f]}, \Gamma, \deltadet, \{\bar{q}_0\}, 2^{\{\bar{q}_0\} \cup Q[f,X_f]} \setminus \{\emptyset\})\] be the determinization of~$\A[f,X_f]$ obtained by the standard subset construction.
Which states are accepting will not actually be relevant. 
Note that every state reachable via a nonempty path from~$\{\bar{q}_0\}$ is of the form $P \times \{X\}$ with $P \subseteq Q$ and $X \in \Gamma$.

Define a BP~$\Bdet$ based on~$\Adet$ as \[\Bdet \ := \ (\Gamma', \mathord{\btran{}'}, \Prob', \{(f,X_f)\})\,,\] where the set of types $\Gamma' \subseteq 2^{Q[f,X_f]}$ is the set of those states in~$\Adet$ that are reachable (in~$\Adet$) from~$\{\bar{q_0}\}$ via a nonempty path (recall that they are of the form $P \times \{X\}$ with $P \subseteq Q$ and $X \in \Gamma$), and
\begin{align*}
X' &\Xhookrightarrow{p}' \deltadet(X',X_1) \cdots \deltadet(X',X_k)
\end{align*}
for all $X' = P \times \{X\} \in \Gamma'$ with $P \ne \emptyset$ and all $X \btran{p} X_1 \cdots X_k$, and $\emptyset \Xhookrightarrow{1}' \emptyset$.
Here is the key lemma of this section:

\begin{restatable}{lemma}{lemBdet} \label{lem:Bdet}
The following statements are equivalent:
\begin{enumerate}
\item[(i)] The probability that some branch of a $\B[X_f]$-tree is accepted by~$\A[f,X_f]$ is positive.
\item[(ii)] The probability that some branch of a $\Bdet$-tree does not have any nodes of type~$\emptyset$ is positive.
\end{enumerate}
\end{restatable}
We prove \cref{lem:Bdet} in \cref{sub:Bdet-proof}.
It will be used in the proof of \cref{thm:coNBA-1} below; but more importantly, \cref{lem:Bdet} is the foundation of \cref{sec:coUBA}.

Given that \cref{lem:Bdet} reflects the key insight of this section, let us comment further.
Considering that condition~(ii) does not mention a notion of acceptance, one might have two concerns at this point:
\begin{enumerate}
\item[(a)] Condition~(ii) does not obviously imply that with positive probability there is even a branch with infinitely many nodes of types containing $(f,X_f)$.
\item[(b)] Even if with positive probability there is such a branch, it is not obvious that such branches would necessarily correspond to branches of~$\B[X_f]$ that are accepted by~$\A[f,X_f]$.
\end{enumerate}
Even for the special case of Markov chains (i.e., every tree has only a single branch), \cref{lem:Bdet} is not at all obvious, and both concerns (a) and~(b) apply.
Indeed, for Markov chains, Courcoubetis and Yannakakis prove a statement related to \cref{lem:Bdet}, namely \cite[Proposition 4.1.4]{CourcoubetisYannakakis95}, with a proof related to ours and dealing explicitly with concern~(b) above.
For the special case of transition systems (i.e., the BP generates exactly one tree), \cref{lem:Bdet} is simple though: consider the branch that follows a cycle around $(f,X_f)$.
For the general case, we need a result on BPs from \cite{12CDK-MFCS}, dealing with concern~(a) above.
The high-level principle behind the proof of \cref{lem:Bdet} is often used in the analysis of Markov chains: if it is possible, infinitely often, to reach a state with a probability bounded away from~$0$, then this state is almost surely reached infinitely often.
See \cref{sub:Bdet-proof} for a full proof of \cref{lem:Bdet}.

We can now derive a PSPACE procedure for the problem $\P(\coNBA) = 1$ without resorting to DPAs:
\begin{restatable}{theorem}{thmcoNBAone} \label{thm:coNBA-1}
The problem $\P(\coNBA) = 1$ is PSPACE-complete.
\end{restatable}

\Cref{thm:NBA-0} (for NBAs) has a coNBA-analogue:

\begin{restatable}{theorem}{thmcoNBAzero} \label{thm:coNBA-0}
The problem $\P(\coNBA) = 0$ is EXPTIME-complete.
It is EXPTIME-hard even for NBAs all whose states are accepting.
\end{restatable}

\section{Co-Unambiguous B\"uchi Automata} \label{sec:coUBA}

In this section we build on the previous section, in particular on \cref{lem:Bdet}, to derive our main technical result: given a BP~$\B$ and an \emph{unambiguous} B\"uchi automaton (UBA)~$\A$, one can decide in~NC whether $\B$ almost surely generates a tree all whose branches are rejected by~$\A$:
\begin{proposition} \label{prop:coUBA-1}
The problem $\P(\coUBA) = 1$ is in~NC.
\end{proposition}
The rest of the section is devoted to the proof of this theorem.
Fix a BP~$\B$ and a UBA~$\A$.
Since NC is closed under complement, we can focus on the problem whether the probability is positive that a $\B$-tree has some branch accepted by~$\A$.
We use \cref{lem:UBA-X1-f}.
Since reachability in a graph is in~NL and, hence, in~NC, it suffices to decide in~NC whether the probability that some branch of a $\B[X_f]$-tree is accepted by~$\A[f,X_f]$ is positive.
By \cref{lem:Bdet} it suffices to decide in~NC whether the probability that some branch of a $\Bdet$-tree does not have any nodes of type~$\emptyset$ is positive.
The challenge is that $\Bdet$ may be exponentially larger than~$\A$, so we need to exploit the unambiguousness of~$\A$ and the regular structure it gives to~$\Bdet$.

Let $\Bdet''$ be the BP (with $\varepsilon$-rules allowed) obtained from $\Bdet$ by removing the type~$\emptyset$ and eliminating all occurrences of type~$\emptyset$ from all right-hand sides.
The probability that a $\Bdet$-tree has an infinite branch of non-$\emptyset$ nodes is equal to the probability that a $\Bdet''$-tree is infinite.
Hence, it remains to show that one can decide in~NC whether the probability that a $\Bdet''$-tree is infinite is positive.

Define a matrix $M \in \Q^{Q[f,X_f] \times Q[f,X_f]}$ whose rows and columns are indexed with the non-$\bar{q}_0$ states of~$\A[f,X_f]$:
\[
 M_{(q,X),(r,Y)} \ := \
\begin{cases} \displaystyle\sum_{X \btran{p} u} p |u|_Y & \text{if } (q,X) \xrightarrow{Y} (r,Y) \ \text{ in } \A[f,X_f] \\[1mm] 0 & \text{otherwise\,,}
\end{cases}
\]
where $|u|_Y \in \N_0$ is the number of occurrences of~$Y$ in~$u$.
(Think of $M_{(q,X),(r,Y)}$ as the expected number of $(r,Y)$-``successors'' of~$(q,X)$.)
The graph of~$M$ is equal to the transition graph of~$\A[f,X_f]$ (excluding~$\bar{q}_0$), which is strongly connected.

Say that $\A[f,X_f]$ has \emph{proper branching} if there exist $(q,Y) \xrightarrow{Z_1} (r_1,Z_1)$ and $(q,Y) \xrightarrow{Z_2} (r_2,Z_2)$ in $\A[f,X_f]$ and a rule $Y \btran{p} u_1 Z_1 u_2 Z_2 u_3$ in~$\B$ with $u_1, u_2, u_3 \in \Gamma^*$.
Now we can state the key lemma:
\begin{restatable}{lemma}{lemkey} \label{lem:key}
Let $\rho$ be the spectral radius of~$M$.
The probability that a $\Bdet''$-tree is infinite is positive if and only if either $\rho > 1$ or $\rho = 1$ and $\A[f,X_f]$ does not have proper branching.
\end{restatable}
Observe the similarity between \cref{lem:key,lem:as-finiteness-char}.
In fact, the proof of \cref{lem:key}, given below, is based on \cref{lem:as-finiteness-char}.
\Cref{lem:key} shows that properties of $\A[f,X_f]$ and $M$ (which are polynomial-sized objects) determine a property of the exponential-sized BP~$\Bdet''$.
Unambiguousness of~$\A[f,X_f]$ is crucial for that connection.

Given that \cref{lem:key} reflects the key insight of this section (if not of this paper), let us comment further.
Suppose $\A[f,X_f]$ has two outgoing transitions in a state $(q,Y)$, say $(q,Y) \xrightarrow{Z_1} (r_1,Z_1)$ and $(q,Y) \xrightarrow{Z_2} (r_2,Z_2)$.
This branching could be ``proper branching'' as defined before \cref{lem:key}, or the original UBA~$\A$ could be nondeterministic when reading~$Y$ in~$q$ and have transitions $q \xrightarrow{Y} r_1$ and $q \xrightarrow{Y} r_2$.
Either type of branching causes non-$0$ entries in the matrix~$M$ and, intuitively, increases its spectral radius~$\rho$.
\Cref{lem:key} tells us that the probability that a $\Bdet''$-tree is infinite is governed by the \emph{combined} effect on~$\rho$ of both types of branching: if $\rho > 1$ then a $\Bdet''$-tree is infinite with positive probability; only in the borderline case, $\rho=1$, the type of branching matters.
Again, this characterization is only correct if the nondeterminism in~$\A$ does not cause ambiguousness.

Let us consider what \cref{lem:key} states for the special case of Markov chains.
In that case, clearly there is no proper branching.
One can show, using unambiguousness, that for Markov chains the spectral radius~$\rho$ of the matrix~$M$ is at most~$1$.
Hence, \cref{lem:key} states for Markov chains that the probability that a $\Bdet''$-tree (consisting of a single branch) is infinite is positive if and only if $\rho = 1$.
Indeed, a related statement can be found in \cite[Lemma~6]{16BKKKMW-CAV}.

To finish the proof of \cref{prop:coUBA-1} it suffices to show that we can check the conditions of \cref{lem:key} in~NC.
Indeed, for comparing the spectral radius with~$1$, we employ \cref{lem:determine-spectral-radius}.
One can check for proper branching in logarithmic space, hence in~NC.
This completes the proof of \cref{prop:coUBA-1}.

\section{LTL} \label{sec:LTL}

With \cref{prop:coUBA-1} from the previous section, we can now show our headline result:

\begin{theorem} \label{thm:LTL-1}
The problem $\P(\LTL)=1$ is PSPACE-complete.
\end{theorem}
\begin{proof}
PSPACE-hardness is immediate in two different ways.
It follows both from the PSPACE-hardness of model checking Markov chains against LTL and from the PSPACE-hardness of model checking transition systems against LTL \cite{SistlaClarke85}.
Both model-checking problems are special cases of $\P(\LTL) = 1$.

Towards membership in PSPACE, there is a classical PSPACE procedure that translates an LTL formula into an (exponential-sized) B\"{u}chi automaton~\cite{VardiWolper86}.
As noted by several authors (e.g., \cite{CouSahSut03,ChaKat14}), this procedure can easily be adapted to ensure that the B\"uchi automaton be a UBA.
By applying this translation to the negation $\neg \varphi$ of the input formula~$\varphi$, we obtain a UBA that rejects exactly those words that satisfy~$\varphi$.
By \cref{prop:coUBA-1} the problem $\P(\coUBA) = 1$ is in~NC.
By \cref{lem:PSPACE-transducer} it follows that $\P(\LTL) = 1$ is in PSPACE.
\end{proof}

Finally we show the following result, exhibiting a big complexity gap between the problems $\P(\LTL)=1$ and $\P(\LTL)=0$.

\begin{restatable}{theorem}{thmLTLzero} \label{thm:LTL-0}
The problem $\P(\LTL)=0$ is 2EXPTIME-complete.
\end{restatable}
\begin{proof}
For membership in 2EXPTIME, we use again the classical procedure that translates an LTL formula into an exponential-sized B\"{u}chi automaton~\cite{VardiWolper86} and then invoke \cref{thm:NBA-0}.

For 2EXPTIME-hardness we adapt the reduction from \cite[Theorem 3.2.1]{CourcoubetisYannakakis95} for MDPs.
The details are in \cref{app:LTL-0}.
\end{proof}

\section{Conclusions} \label{sec:conclusions}

We have devised a PSPACE procedure for $\P(\LTL) = 1$, i.e., qualitative LTL model checking of BPs.
The best previously known procedure ran in 2EXPTIME~\cite{12CDK-MFCS}.
Since BPs naturally generalize both transition systems and Markov chains (for both of which LTL model checking is PSPACE-complete), one might view our model-checking algorithm as an optimal general procedure.
The same holds for NBA-specifications instead of LTL.

The main technical ingredients have been the automata-theoretic approach and the algorithmic analysis of UBAs, nonnegative matrices, and finiteness of BPs.
Our proofs were inspired by the observation that the spectral radii of certain nonnegative matrices are central to model checking Markov chains against UBAs, and also determine fundamental properties of BPs.
Very loosely speaking, when model checking Markov chains against UBAs, the spectral radius measures the amount of nondeterministic branching in the UBA, whereas when analyzing BPs, the spectral radius measures the amount of tree branching.
The ``general case'', i.e., model checking BPs, features both kinds of branching.
Serendipitously, an analysis of spectral radii still leads, as we have seen, to optimal algorithms.

We have also established the complexities of related problems, partially as a tool for the mentioned LTL and NBA problems and partially to map out the landscape.
We have shown that the $\P(\cdot) = 0$ variants are more complex than their $\P(\cdot) = 1$ counterparts.
An intuitive explanation of this phenomenon is that for an instance of an $\P(\cdot)=1$ problems to be negative, tree branching and probabilistic branching ``work together'' to falsify the specification on some branch.
In contrast, for $\P(\cdot)=0$ problems, tree branching and probabilistic branching are ``adversaries'', like in MDPs.
Indeed, for lower bounds on $\P(\cdot)=0$ problems we have encoded alternation in various forms.

One might ask about the complexity of $\P(\UBA) = 1$.
Indeed, in trying to solve $\P(\LTL) = 1$ efficiently, the authors set out to solve $\P(\UBA)=1$ efficiently (perhaps in P or even NC), with the PSPACE transduction from LTL to UBA in mind.
However, the complexity of UBA universality is an open problem \cite{Rabinovich18}; only membership in PSPACE is known.
So even for the fixed transition system with $a \btran{1} a b$ and $b \btran{1} a b$ the problem $\P(\UBA)=1$ cannot be placed in P without improving the complexity of UBA universality.
A PSPACE-hardness proof of $\P(\UBA)=1$ might have to make use of both types of branching in BPs, as $\P(\UBA)=1$ is in NC for Markov chains~\cite{16BKKKMW-CAV}.

Model checking BPs quantitatively, i.e., computing the satisfaction probability, comparing it with a threshold, or approximating it, is left for future work.
Exact versions of these problems are computationally complex, as they are at least as hard as the corresponding $\P(\cdot)=0$ problem.
The paper~\cite{12CDK-MFCS} describes, for DPAs, nonlinear equation systems whose least nonnegative solution characterizes the satisfaction probabilities.
Newton's method is efficient for approximating the solution of such equation systems; see \cite{StewartEY15,EtessamiSY17}.

\bibliography{literature}

\newpage
\appendix
\crefalias{section}{appendix}
\section{Proofs of \texorpdfstring{\cref{sec:basic-results}}{Section \ref{sec:basic-results}}}
\subsection{Proof of \texorpdfstring{\cref{lem:as-finiteness-char}}{Lemma \ref{lem:as-finiteness-char}}}
\lemasfinitenesschar*
\begin{proof}
By \cref{lem:as-finite-SCC} it suffices to show that if $G$ is strongly connected then a $\B$-tree is infinite with positive probability if and only if $\Gamma$ is supercritical or linear.
So, let $G$ be strongly connected.

Call a type $X \in \Gamma$ \emph{immortal} if there is no finite $X$-tree.
If a type~$X$ is immortal, then the probability that a $\B[X]$-tree is finite is~$0$, and for all $Y \in \Gamma$ the probability that a $\B[Y]$-tree is finite is less than~$1$ (as $G$ is strongly connected).
If $\Gamma$ is linear then all types are immortal.
So we can assume in the following that $\Gamma$ is not linear.

Since $\Gamma$ is not linear and $G$ is strongly connected, for all $X \in \Gamma$ there is a finite prefix of an $X$-tree that has at least two leaves of type~$X$; and an $X$-tree has, with positive probability, that finite prefix.
Suppose some type, say $X$, is immortal.
Then every $X$-tree that has a finite prefix with $k$ ($k \in \N$) leaves of type~$X$ has at least $k$ (infinite) branches that go through those $k$ nodes of type~$X$.
By strong connectedness, it follows that an $X$-tree that has a finite prefix with $k$ leaves of type~$X$ has, with probability~$1$ (i.e., $\P_X$-almost surely), a finite prefix with $k+1$ leaves of type~$X$.
Hence, with probability~$1$ we have $\lim_{i \to \infty} Z_i = \infty$, where $Z_i$ is the number of nodes that are exactly $i$ levels under the root of an $X$-tree.
Denoting by $\Ex$ the expectation with respect to~$\P_X$, it follows with Fatou's lemma that $\lim_{i \to \infty} \Ex Z_i = \infty$.
We can write $\Ex Z_i = (M^i \vec{1})_X$, where $\vec{1}$ denotes the vector in~$\{1\}^\Gamma$.
Then \[ \lim_{i \to \infty} \norm{M^i \vec{1}} \ = \ \infty\,,\] where $\norm{\cdot}$ denotes the $1$-norm.
Since $G$ is strongly connected, by the Perron-Frobenius theorem \cite[Theorem~2.1.4~(b)]{BermanP94}, $M$ has an eigenvector $v \in \R^\Gamma$ with $v_Y > 0$ for all $Y \in \Gamma$ and $M v = \rho v$, where $\rho$ denotes the spectral radius of~$M$.
Since $M^i v = \rho^i v$, it follows that \[\lim_{i \to \infty} \rho^i \norm{v} \ = \ \lim_{i \to \infty} \norm{M^i v} \ = \ \infty\,.\]
Hence $\rho>1$.
We conclude that if some type is immortal, then we have $\rho>1$ and the probability that a $\B$-tree is finite is less than~$1$.

On the other hand, if no type is immortal, then it follows from \cite[Section~3]{13EGK-IPL} (building on \cite[Section~8.1]{EtessamiYannakakis09}) that a $\B$-tree is infinite with positive probability if and only if $\rho > 1$.
We conclude that, regardless of whether there exists an immortal type, a $\B$-tree is infinite with positive probability if and only if $\Gamma$ is supercritical.
\end{proof}

\subsection{Proof of \texorpdfstring{\cref{prop:as-finite-NC}}{Proposition \ref{prop:as-finite-NC}}}
\propasfiniteNC*
\begin{proof}
We use the characterization from \cref{lem:as-finiteness-char}.
One can compute in NL, and hence in NC, the directed acyclic graph of SCCs of~$G$ and the set of types that are reachable from~$X_0$.
Therefore, we assume without loss of generality that $G$ is strongly connected.
By \cref{lem:determine-spectral-radius} one can check in NC whether $\Gamma$ is supercritical.
Whether $\Gamma$ is linear can be checked in logarithmic space, hence in~NC.
\end{proof}

\subsection{Proof of \texorpdfstring{\cref{lem:AFT-1}}{Lemma \ref{lem:AFT-1}}}
\lemAFTone*
\begin{proof}
The problem can be rephrased as almost-sure finiteness of BPs with $\varepsilon$-rules allowed.
Indeed, given $\B$ and~$T$, one can eliminate all occurrences of types in~$T$ from all right-hand sides; then, $\finally T$ in the original BP corresponds to finiteness in the new BP.
Hence, the result follows from \cref{prop:as-finite-NC}.
\end{proof}

\subsection{Proof of \texorpdfstring{\cref{thm:DPA-1}}{Theorem \ref{thm:DPA-1}}}
\thmDPAone*
\begin{proof}
In \cite[Section~3]{12CDK-MFCS} it was shown that the problem $\P(\DPA) = 1$ can be decided in polynomial time.
We show how to implement this approach in NC.
In \cite[Section~3]{12CDK-MFCS} a product of the BP and the DPA is computed and analyzed; the product is a BP whose types are coloured with \emph{priorities} (natural numbers).
The product, call it $\B = \BPtuple$, can be computed in logarithmic space.
The question is then whether the probability is~$1$ that all branches of a $\B$-tree are such that the highest priority that appears infinitely often is even.
It is shown in \cite[Section~3.1]{12CDK-MFCS} that this is the case if and only if all types $X$ that are reachable from~$X_0$ and are associated with an odd priority satisfy $\P_X(\finally N_X) = 1$, where $N_X \subseteq \Gamma$ is a certain set of types that can be computed in NL by a simple reachability analysis in~$\B$.
By \cref{lem:AFT-1} one can determine in NC whether $\P_X(\finally N_X) = 1$.
The theorem follows.
\end{proof}

\subsection{Proof of \texorpdfstring{\cref{thm:DPA-0}}{Theorem \ref{thm:DPA-0}}}
\thmDPAzero*
\begin{proof}
Membership in P was shown in \cite[Theorem~15]{12CDK-MFCS}.
For P-hardness we reduce from the monotone circuit value problem.
Given a monotone circuit with output gate $g_{\mathit{out}}$, we construct a BP $\B = (\Gamma,\mathord{\btran{}},\Prob,X_{g_{\mathit{out}}})$ and a DBA~$\A$ such that $\P_{X_{g_{\mathit{out}}}}(\A \text{ accepts}) > 0$ if and only if $g_{\mathit{out}}$ evaluates to~$1$.

For each $\land$- and each $\lor$-gate~$g$ include a type $X_g \in \Gamma$.
Also include types $X_0, X_1 \in \Gamma$ for the inputs $0, 1$, respectively.
For each $\land$-gate~$g$ with children $g_1, \ldots, g_k$ include a rule
\begin{align*}
X_g &\btran{1} X_{g_1} \cdots X_{g_k}\,.
\intertext{For each $\lor$-gate~$g$ with children $g_1, \ldots, g_k$ include rules}
X_g &\btran{1/k} X_{g_1},\ \ldots,\ X_g \btran{1/k} X_{g_k}\,.
\intertext{Include rules}
X_0 &\btran{1} X_0 \text{ and }X_1 \btran{1} X_1\,.
\end{align*}

Construct a two-state DBA~$\A$ that accepts exactly those $w \in \Gamma^\omega$ that contain~$X_1$. 
It follows from a straightforward induction over the gate height (longest distance to an input) that, for any gate~$g$, it evaluates to~$1$ if and only if $\P_{X_g}(\A \text{ accepts}) > 0$.
\end{proof}

\input{app-defs}
\input{app-NBA-0}

\section{Proofs of \texorpdfstring{\cref{sec:coNBA}}{Section \ref{sec:coNBA}}}
\input{app-coNBA-1}

\subsection{Proof of \texorpdfstring{\cref{lem:Bdet}}{Lemma~\ref{lem:Bdet}}} \label{sub:Bdet-proof}

In this subsection we prove \cref{lem:Bdet}, which is instrumental for the main results of the paper.

\lemBdet*
\medskip

Fix a word $w \in \Gamma^+$ such that $\delta[f,X_f]((f,X_f),w)$ is maximal, i.e., there is no $w' \in \Gamma^+$ such that $\delta[f,X_f]((f,X_f),w') \supsetneq \delta[f,X_f]((f,X_f),w)$.

\begin{lemma} \label{lem:UBA-w}
Let $v \in \Gamma^*$ be a word such that $\delta[f,X_f]((f,X_f),v) \ni (f,X_f)$.
Then, $\delta[f,X_f]((f,X_f),v w) = \delta[f,X_f]((f,X_f), w)$; i.e., for any path $(f,X_f) \Xrightarrow{v w}* (q,X)$ there is a path $(f,X_f) \Xrightarrow{v}* (f,X_f) \Xrightarrow{w}* (q,X)$.
\end{lemma}
\begin{proof}
Since $\delta[f,X_f]((f,X_f),v) \ni (f,X_f)$, we have \[\delta[f,X_f]((f,X_f),v w) \ \supseteq \ \delta[f,X_f]((f,X_f), w)\,.\]
But $w$ is maximal.
\end{proof}

We enrich~$\Adet$ to obtain a DBA, $\Adet'$, whose states have an additional component keeping track of whether the word~$w \in \Gamma^+$ from above is being seen.
The accepting runs of $\Adet'$ contain infinitely many $w$-labelled segments that, loosely speaking, ``start from $(f,X_f)$''.
Formally, let \[W = \{0\} \cup \{v \in \Gamma^* \mid v \text{ is a suffix of } w\}\,,\] including~$w$ and the empty word~$\varepsilon$.
We assume $0 \not\in \Gamma$.
Define
\[
\Adet' \ :=\ (2^{\{\bar{q}_0\} \cup Q[f,X_f]} \times W, \deltadet', (\{\bar{q}_0\},0), (2^{Q[f,X_f]} \setminus \{\emptyset\}) \times \{\varepsilon\})\,,
\]
where
\begin{align*}
\deltadet'((U,X v), X) \ &= \ (\deltadet(U,X), v) \\
\deltadet'((U,X v), Y) \ &= \ (\deltadet(U,X), 0) \quad &&\text{for } X \ne Y \\
\deltadet'((U,v), X)   \ &= \ (\deltadet(U,X), w) \quad &&\text{for } v \in \{0, \varepsilon\},\ (f,X_f) \in \deltadet(U,X) \\
\deltadet'((U,v), X)   \ &= \ (\deltadet(U,X), 0) \quad &&\text{for } v \in \{0, \varepsilon\},\ (f,X_f) \not\in \deltadet(U,X)\,.
\end{align*}
It follows from \cref{lem:UBA-w} and the construction of~$\Adet'$ that $\Adet'$ has a single accepting state reachable from~$(\{\bar{q}_0\},0)$, namely $(\deltadet(\{(f,X_f)\},w), \varepsilon)$.

\begin{lemma} \label{lem:UBA-3-equivalences}
The following statements are equivalent:
\begin{enumerate}
\item[(i)] The probability that some branch of a $\B[X_f]$-tree is accepted by~$\A[f,X_f]$ is positive.
\item[(ii)] The probability that some branch of a $\B[X_f]$-tree has a run (accepting or not) in~$\A[f,X_f]$ is positive.
\item[(iii)] The probability that some branch of a $\B[X_f]$-tree is accepted by~$\Adet'$ is positive.
\end{enumerate}
\end{lemma}

\Cref{lem:UBA-3-equivalences} implies \cref{lem:Bdet}, as it follows from the definition of~$\Bdet$ that the probability that some branch of a $\B[X_f]$-tree has a run in~$\A[f,X_f]$ (cf.\ condition~(ii) in \cref{lem:UBA-3-equivalences}) equals the probability that some branch of a $\Bdet$-tree does not have any nodes of type~$\emptyset$ (cf.\ condition~(ii) of \cref{lem:Bdet}).
So it remains to prove \cref{lem:UBA-3-equivalences}.

\begin{proof}[Proof of \cref{lem:UBA-3-equivalences}]
(i) $\Longrightarrow$ (ii).
Trivial.

(iii) $\Longrightarrow$ (i).
Let $X_f X_1 X_2 \cdots$ be accepted by~$\Adet'$.
Then $X_1 X_2 \cdots$ can be decomposed in $v_1 w v_2 w \cdots$ with $v_1, v_2, \ldots \in \Gamma^*$ such that $\A[f,X_f]$ has paths \[(f,X_f) \Xrightarrow{v_1 w v_2 w \cdots v_i}* (f,X_f)\] for all~$i \ge 1$.
Let $i \ge 2$.
There is $(q,X)$ with \[(f,X_f) \Xrightarrow{v_1 w v_2 w \cdots v_{i-1} w}* (q,X) \Xrightarrow{v_i}* (f,X_f)\,.\]
By \cref{lem:UBA-w} there is a path \[(f,X_f) \Xrightarrow{v_1 w v_2 w \cdots v_{i-1}}* (f,X_f) \Xrightarrow{w}* (q,X)\,.\]
Thus also $(f,X_f) \Xrightarrow{w v_i}* (f,X_f)$.
Since $i \ge 2$ was arbitrary, it follows that $\A[f,X_f]$ has an accepting run \[\bar{q}_0 \xrightarrow{X_f} (f,X_f) \Xrightarrow{v_1}* (f,X_f) \Xrightarrow{w v_2}* (f,X_f) \Xrightarrow{w v_3}* \cdots\,.\]

(ii) $\Longrightarrow$ (iii).
For this part we use results from~\cite{12CDK-MFCS}, which considers the problem of model checking BPs against DPAs.
We need only a special case of such automata: (a)~the same (word) automaton is run on every branch of the tree, and (b)~our automaton~$\Adet'$ is a DBA, which can be viewed as a DPA whose states are labelled with only $2$~priorities: priority~$0$ for non-accepting states and priority~$1$ for accepting states.
The paper~\cite{12CDK-MFCS} constructs a product BP from the BP and the automaton, and subsequently considers BPs whose types are coloured with a priority.
In this way the model-checking problem reduces to computing the probability that there is a branch on which the highest priority that occurs infinitely often is odd.
Since our automaton~$\Adet'$ already embeds a BP, instead of taking another product, we define a BP, $\Bdet'$,  more directly based on~$\Adet'$, in the same way that $\Bdet$ was defined based on~$\Adet$ in \cref{sub:Bdet}.
More explicitly,
\[
\Bdet' \ := \ (\Gamma', \mathord{\btran{}'}, \Prob', (\{(f,X_f)\},w))\,,
\]
where the set of types $\Gamma' \subseteq 2^{Q[f,X_f]} \times W$ is the set of those states in~$\Adet'$ that are reachable (in~$\Adet'$) from~$(\{\bar{q}\},0)$ via a nonempty path (recall that they are of the form $(P \times \{X\},v)$ with $P \subseteq Q$ and $X \in \Gamma$ and $v \in W$), and
\begin{align*}
X' &\Xhookrightarrow{p}' \deltadet'(X',X_1) \cdots \deltadet'(X',X_k)
\intertext{
for all $X' = (P \times \{X\},v) \in \Gamma'$ with $P \ne \emptyset$ and all $X \btran{p} X_1 \cdots X_k$, and
}
(\emptyset,v) &\Xhookrightarrow{1}' (\emptyset,v)
\end{align*}
for all $v \in W$.
Recall that \[(\deltadet(\{(f,X_f)\},w), \varepsilon) =: X_w'\] is the single accepting state in~$\Adet'$ that is reachable from $(\{\bar{q}_0\},0)$.
We call a branch of~$\Bdet'$ \emph{accepting} if it contains $X_w'$ infinitely often.

Suppose ``(ii)'', i.e., the probability that some branch of a $\B[X_f]$-tree has a (non-accepting or accepting) run in~$\A[f,X_f]$ is positive.
Thus, the probability that some branch of a $\B[X_f]$-tree has a run in~$\Adet$ that does not enter the state~$\emptyset$ is positive.
Let $X_1 w' \in \Gamma^+$ be such that \[(f,X_f) \in \deltadet(\{(f,X_f)\},w X_1 w')\,.\]
Then, also the probability that some branch of a $\B[X_f]$-tree starts with $X_f w X_1 w'$ and has a run in~$\Adet$ that does not enter~$\emptyset$ is positive.
Thus, the probability that some branch of a $\B[X_1]$-tree starts with $X_1 w'$ and has a run in~$\Adet$, started in $\deltadet(\{(f,X_f)\},w)$, that does not enter~$\emptyset$ is positive.
Hence, the probability that some branch of a $\B[X_1]$-tree has a run in~$\Adet'$, started in $(\deltadet(\{(f,X_f)\},w), \varepsilon)$, that does not enter a state of the form $(\emptyset,v)$ is positive.
From the construction of~$\Bdet'$ it follows that the probability that some branch of a $\Bdet'[X_w']$-tree (i.e., with $X_w' = (\deltadet(\{(f,X_f)\},w), \varepsilon)$ as start type) does not have a node of a type of the form $(\emptyset,v)$ is positive.

Consider any type $(U_0,v_0)$ in~$\Bdet'$ with $U_0 \ne \emptyset$ and view it as a state in~$\Adet'$.
Then there is $u_1 \in \Gamma^*$ with \[(U_0,v_0) \Xrightarrow{u_1}* (U_1,v_1)\] where $U_1 \ne \emptyset$ and $v_1 \in \{0, \varepsilon\}$.
Let $u_2 \in \Gamma^*$ be a shortest word such that there is $(q_1,X_1) \in U_1$ with \[(q_1,X_1) \Xrightarrow{u_2}* (f,X_f)\] in~$\A[f,X_f]$.
It follows that in~$\Adet'$ we have \[(U_0,v_0) \Xrightarrow{u_1 u_2}* (U_2,w) \Xrightarrow{w}* (U_3,\varepsilon)\] with $(f,X_f) \in U_2$.
By \cref{lem:UBA-w} we have \[U_3 = \deltadet(\{(f,X_f)\},w)\,.\]
Hence $(U_0,v_0) \Xrightarrow{}* X_w'$.

Combining this reachability fact with the previous argument, we infer that the probability that some branch of a $\Bdet'[X_w']$-tree has only nodes of types from which $X_w'$ is reachable (in~$\Bdet'$) is positive.
It follows from \cite[Lemma~11]{12CDK-MFCS} that the probability that some branch of a $\Bdet'[X_w']$-tree is accepting is positive.%
\footnote{In terms of the notation therein, we instantiate \cite[Lemma~11]{12CDK-MFCS} with $X := X_w'$.
By the reachability argument above, $N_X = N_{X_w'}$ does not include types of the form $(U,v)$ with $U \ne \emptyset$.
Thus, we have argued that the probability of $\finally N_{X_w'}$ is $p < 1$.
Hence, \cite[Lemma~11]{12CDK-MFCS} asserts that the probability that a $\Bdet'[X_w']$-tree has an \emph{$X_w'$-branch} equals $1-p > 0$, where \emph{$X_w'$-branch} means accepting path in terms of our definition.
}
Hence, the probability that some branch of a $\Bdet'$-tree is accepting is positive.
From the construction of~$\Bdet'$ it follows that the probability that some branch of a $\B[X_f]$-tree is accepted by~$\Adet'$ is positive, i.e., ``(iii)''.
\end{proof}

\subsection{Proof of \texorpdfstring{\cref{thm:coNBA-1}}{Theorem \ref{thm:coNBA-1}}}
\thmcoNBAone*
\begin{proof}
Towards membership in PSPACE, fix a BP~$\B$ and an NBA~$\A$.
Since PSPACE is closed under complement, we can focus on the problem whether the probability is positive that a $\B$-tree has a branch accepted by~$\A$.
We use \cref{lem:UBA-X1-f}.
Since reachability in a graph is in~NL and, hence, in PSPACE, it suffices to decide in PSPACE whether the probability that some branch of a $\B[X_f]$-tree is accepted by~$\A[f,X_f]$ is positive.
In order to check that, by \cref{lem:Bdet} it suffices to construct the BP~$\Bdet$ and then invoke the NC procedure of \cref{lem:AFT-1} to check if $\P_{\Bdet}(\finally \{\emptyset\}) < 1$.
The BP~$\Bdet$ has exponential size but can be computed with a PSPACE transducer.
(In particular, whether a state in~$\Adet$ is reachable from~$\{\bar{q}\}$ via a nonempty path can be determined in NPSPACE $=$ PSPACE.)
By \cref{lem:PSPACE-transducer} it follows that $\P(\coNBA) = 1$ is in PSPACE.

PSPACE-hardness follows from the PSPACE-hardness~\cite{Vardi85} of the \emph{probabilistic emptiness problem}, which, given a Markov chain and an NBA, asks if the probability is~$0$ that the Markov chain generates a word accepted by the NBA.
(We remark that, in contrast, the problem whether a given transition system has a run accepted by a given NBA is in NL: search the product for an accepting cycle).
\end{proof}

\input{app-coNBA-0}

\section{Proof of \texorpdfstring{\cref{lem:key}}{Lemma \ref{lem:key}}}
\lemkey*
\begin{proof}
The automaton $\A[f,X_f]$ is unambiguous, as $\A$~is unambiguous, and has $(f,X_f)$ as (the only) accepting state.
Recall also that $(f,X_f)$ is reachable from all states in~$\A[f,X_f]$.
It follows that $\A[f,X_f]$ does not have \emph{diamonds}, i.e., $\A[f,X_f]$ does not have states $(q,X), (q',X')$ and a word $u \in \Gamma^*$ such that $\A[f,X_f]$ has two different paths $(q,X) \Xrightarrow{u}* (q',X')$.

By the Perron-Frobenius theorem \cite[Theorem~2.1.4~(b)]{BermanP94}, $M$ has an eigenvector $v \in (0,\infty)^{Q[f,X_f]}$ (all entries positive) with $M v = \rho v$.
(Think of the entries of~$v$ as ``weights'' of the states in~$\A[f,X_f]$.
Loosely speaking, the equality $(M v)_{(q,X)} = \rho v_{(q,X)}$ expresses that the expected combined weight of the ``successors'' of~$(q,X)$ is equal to the weight of~$(q,X)$ multiplied by~$\rho$.)

We ``lift''~$v$ to define a vector $\bar{v} \in [0,\infty)^{\Gamma'}$ where $\Gamma'$ is the set of types in~$\Bdet$ (recall that they are of the form $P \times \{X\}$ with $P \subseteq Q$ and $X \in \Gamma$):
\[
 \bar{v}_{P \times \{X\}} \ := \ \sum_{q \in P} v_{(q,X)}
\]
Denote by $\bar{M} \in \Q^{\Gamma' \times \Gamma'}$ the matrix defined before \cref{lem:as-finiteness-char}, but for $\Bdet$.
Then we have for all $P \times \{X\} \in \Gamma'$:
\begin{align*}
    (\bar{M} \bar{v})_{P \times \{X\}}
&\ =\ \sum_{X \btran{p} X_1 \cdots X_k} p \sum_{i=1}^k \bar{v}_{\deltadet(P \times \{X\}, X_i)} \\
&\ =\ \sum_{X \btran{p} X_1 \cdots X_k} p \sum_{i=1}^k \ \sum_{(r,X_i) \in \deltadet(P \times \{X\}, X_i)} v_{(r,X_i)} \\
&\ \mathop{=}^*\ \sum_{X \btran{p} X_1 \cdots X_k} p \sum_{i=1}^k \ \sum_{q \in P} \ \sum_{r : (q,X) \xrightarrow{X_i} (r,X_i)} v_{(r,X_i)} \\
&\ =\ \sum_{q \in P} \ \sum_{X \btran{p} X_1 \cdots X_k}  p \sum_{i=1}^k \ \sum_{r : (q,X) \xrightarrow{X_i} (r,X_i)} v_{(r,X_i)} \\
&\ =\ \sum_{q \in P} (M v)_{(q,X)}\\
&\ =\ \sum_{q \in P} \rho v_{(q,X)}\\
&\ =\ \rho \bar{v}_{P \times \{X\}}\;,
\end{align*}
where the third equality (marked with $\displaystyle\mathop{=}^*$) holds as for any $(r,X_i) \in \deltadet(P \times \{X\}, X_i)$ there is exactly one $q \in P$ with $(q,X) \xrightarrow{X_i} (r,X_i)$ in~$\A[f,X_f]$.
Indeed, towards a contradiction, suppose there are $q_1, q_2 \in Q$ with $q_1 \ne q_2$ and $(q_j,X) \xrightarrow{X_i} (r,X_i)$ for both $j \in \{1,2\}$.
Since $P \times \{X\}$ is reachable in~$\Adet$ from $\{(f,X_f)\}$, there is $u \in \Gamma^*$ with \[(f,X_f) \Xrightarrow{u}* (q_j,X) \xrightarrow{X_i} (r,X_i)\] in~$\A[f,X_f]$ for both $j \in \{1,2\}$, contradicting the absence of diamonds in~$\A[f,X_f]$.
We conclude from the above computation that $\bar{M} \bar{v} = \rho \bar{v}$; i.e., $\bar{v}$ is an eigenvector of~$\bar{M}$ with eigenvalue~$\rho$.

In the following, for subsets $\Delta \subseteq \Gamma'$, we write $\bar{M}_{\Delta} \in \Q^{\Delta \times \Delta}$ for the (square) principal submatrix obtained from~$\bar{M}$ by restricting it to the rows and columns indexed by elements of~$\Delta$.
Similarly, define $\bar{v}_{\Delta} \in [0,\infty)^{\Delta}$  by restricting $\bar{v}$ to the entries indexed by elements of~$\Delta$.
Since $\bar{M}$ and~$\bar{v}$ are nonnegative, we have $\bar{M}_{\Delta} \bar{v}_{\Delta} \le \rho \bar{v}_{\Delta}$ (the inequality is meant componentwise).

Note that $\bar{v}_\emptyset = 0$.
Define $\Gamma'' := \Gamma' \setminus \{\emptyset\}$.
Thus we have $\bar{M}_{\Gamma''} \bar{v}_{\Gamma''} = \rho \bar{v}_{\Gamma''}$.
All entries of~$\bar{v}_{\Gamma''}$ are positive, as all entries of~$v$ are.
By Perron-Frobenius theory \cite[Corollary~2.1.12]{BermanP94} it follows that $\rho$ is the spectral radius of~$\bar{M}_{\Gamma''}$.
The matrix $\bar{M}_{\Gamma''}$ is equal to the matrix defined before \cref{lem:as-finiteness-char}, but for~$\Bdet''$.
We complete the proof with the following case distinction.
\begin{itemize}
\item
Suppose $\rho > 1$.
Let $\Delta \subseteq \Gamma''$ be a bottom SCC of the graph of~$\bar{M}_{\Gamma''}$.
As $\Delta$ is bottom, $\bar{M}_\Delta \bar{v}_\Delta = \rho \bar{v}_\Delta$.
So the spectral radius of~$\bar{M}_\Delta$ is at least~$\rho > 1$ (in fact, it must be equal to~$\rho$).
It follows that $\Delta$~is supercritical in~$\Bdet''$.
Thus, by \cref{lem:as-finiteness-char}, a $\Bdet''$-tree is infinite with positive probability.
\item
Suppose that $\rho = 1$ and that $\A[f,X_f]$ does not have proper branching.
Let $\Delta \subseteq \Gamma''$ be a bottom SCC of the graph of~$\bar{M}_{\Gamma''}$.
Then $\bar{M}_\Delta \bar{v}_\Delta = \bar{v}_\Delta$, so by the Perron-Frobenius theorem \cite[Theorem~2.1.4~(b)]{BermanP94} the spectral radius of~$\bar{M}_\Delta$ is~$1$.
By the absence of proper branching we also have $\bar{M}_\Delta \vec{1} \le \vec{1}$, where $\vec{1}$ denotes the all-$1$ vector, i.e., the element of $\{1\}^\Delta$.
By Perron-Frobenius theory \cite[Theorem~2.1.11]{BermanP94} it follows that $\bar{M}_\Delta \vec{1} = \vec{1}$.
Thus, $\Delta$~is linear in~$\Bdet''$.
Hence, by \cref{lem:as-finiteness-char}, a $\Bdet''$-tree is infinite with positive probability.
\item
Suppose that $\rho=1$ and that $\A[f,X_f]$ has proper branching, i.e., there exist \[(q,Y) \xrightarrow{Z_1} (r_1,Z_1) \text{ and } (q,Y) \xrightarrow{Z_2} (r_2,Z_2) \text{ in } \A[f,X_f]\] and a rule $Y \btran{p} u_1 Z_1 u_2 Z_2 u_3$ with $u_1, u_2, u_3 \in \Gamma^*$.
Consider any SCC $\Delta \subseteq \Gamma''$ of the graph of~$\bar{M}_{\Gamma''}$.
Denote by $\rho_\Delta$ the spectral radius of~$\bar{M}_\Delta$.
As $\bar{M}_\Delta$ is a principal submatrix of~$\bar{M}$, we have $\rho_\Delta \le \rho = 1$ \cite[Corollary 2.1.6~(a)]{BermanP94}.
So $\Delta$ is not supercritical.
\begin{enumerate}
\item[(i)] $\rho_\Delta < 1$.
We have argued before \cref{lem:as-finiteness-char} that $\Delta$ being linear would imply $\rho_\Delta=1$.
Hence, $\Delta$~is not linear.
\item[(ii)] $\rho_\Delta = 1$.
Recall that $\bar{M}_{\Delta} \bar{v}_{\Delta} \le \bar{v}_{\Delta}$, so by Perron-Frobenius theory \cite[Theorem~2.1.11]{BermanP94} we must have $\bar{M}_{\Delta} \bar{v}_{\Delta} = \bar{v}_{\Delta}$.
Let $P \times \{X\} \in \Delta$ and $p \in P$.
Let $u_0 \in \Gamma^*$ with $(p,X) \Xrightarrow{u_0}* (q,Y)$ in~$\A[f,X_f]$. 
Hence \[(p,X) \Xrightarrow{u_0}* (q,Y) \xrightarrow{Z_j} (r_j,Z_j)\] in~$\A[f,X_f]$ for both $j \in \{1,2\}$.
Towards a contradiction, suppose $\Delta$ is linear.
If $\deltadet(P \times \{X\}, u_0)$ is not in~$\Delta$, then neither are $\deltadet(P \times \{X\}, u_0 Z_1)$ or $\deltadet(P \times \{X\}, u_0 Z_2)$.
Otherwise (i.e., $\deltadet(P \times \{X\}, u_0) \in \Delta$) there is $j \in \{1,2\}$ such that $\deltadet(P \times \{X\}, u_0 Z_j) \not\in \Delta$, as $\Delta$ is linear.
Either way there is $j \in \{1,2\}$ such that \[\deltadet(P \times \{X\}, u_0 Z_j) \not\in \Delta\,.\]
Note that $(r_j,Z_j) \in \deltadet(P \times \{X\}, u_0 Z_j) \ne \emptyset$.
Let $x Z$ (with $x \in \Gamma^*$ and $Z \in \Gamma$) be the shortest prefix of~$u_0 Z_j$ such that $U := \deltadet(P \times \{X\}, x) \in \Delta$ but $\deltadet(U,Z) \not\in \Delta$.
Then we have:
\begin{align*}
      \bar{v}_U
&\ =\    (\bar{M}_\Delta \bar{v}_\Delta)_U \\
&\ <\    (\bar{M}_\Delta \bar{v}_\Delta)_U + \bar{M}_{U,\deltadet(U,Z)} \bar{v}_{\deltadet(U,Z)} \\
&\ \le\  (\bar{M} \bar{v})_U \\
&\ =\    \bar{v}_U\,,
\end{align*}
a contradiction.
Thus, $\Delta$ is not linear.
\end{enumerate}
We conclude that in both cases $\Delta$~is neither linear nor supercritical.
Since $\Delta$ was an arbitrary SCC, we conclude from \cref{lem:as-finiteness-char} that a $\Bdet''$-tree is almost surely finite.
\item
Suppose $\rho < 1$.
Then, for any SCC $\Delta$ the spectral radius of $\bar{M}_{\Delta}$ is also less than~$1$ \cite[Corollary 2.1.6~(a)]{BermanP94}, so $\Delta$ is neither supercritical nor linear.
Thus, by \cref{lem:as-finiteness-char}, a $\Bdet''$-tree is almost surely finite.
\end{itemize}
Hence, the probability that a $\Bdet''$-tree is infinite is positive if and only if either $\rho > 1$ or $\rho = 1$ and $\A[f,X_f]$ does not have proper branching.
\end{proof}

\input{app-LTL-0}

\end{document}

%% file: app-defs.tex
\subsection{Additional Definitions Concerning Alternating Turing Machines} \label{app:defs}

An \emph{alternating Turing machine} is a 6-tuple $(S_\exists, S_\forall, \Sigma, T, s_0, \sacc)$, where $S = S_\exists \cup S_\forall \cup \{\sacc\}$ is a finite set of (control) states partitioned into existential states $S_\exists$ and universal states $S_\forall$ and the (only) accepting state~$\sacc$, $\Sigma$ is a finite alphabet, $T \subseteq (S_\exists \cup S_\forall) \times \Sigma \times \Sigma \times \{-1,+1\} \times S$ is a transition relation, and $s_0$ is the initial state.
A transition $(s,a,a',D,s') \in T$ means that if $M$ is in state~$s$ and its head reads letter~$a$, then it rewrites the content of the current cell with the letter~$a'$, it moves the head in direction~$D$ (either left if $D=-1$, or right if $D=+1$), and it changes its state to~$s'$.
We assume that for all $s \in S_\exists \cup S_\forall$ and $a \in \Sigma$ there is at least one outgoing transition.
A configuration of an alternating Turing machine is given by a 3-tuple $(i, s, w)$ where $i \in \N$ indicates the header position, $s \in S$ is the current state of the Turing machine, and $w \in \Sigma^*$ is the contents of the memory tape.

For a word $w \in \Sigma^*$ we will write $w(i)$ to denote its $i$'th component, and $w[i=a]$ to denote the string $w(1)\ldots w(i-1) a w(i+1)\ldots w(n)$. The initial configuration of an alternating Turing machine on a word $w \in \Sigma^*$ is $(1, s_0, w)$, and given a transition $t = (s, a, a', D, s')$ and a configuration $c = (i, s'', w)$ we will use $c \triangleright t$ to denote the configuration $(i+D, s', w[i=a'])$ if $s'' = s$ and $w_i = a$ (and leave it undefined otherwise). We extend this notation to sequences of transitions $t_1\ldots t_n \in T^*$ by writing $c \triangleright t_1\ldots t_n = (c \triangleright t_1) \triangleright t_2\ldots t_n$. We will use $\pi_\N$, $\pi_S$, and $\pi_{\Sigma^*}$ to denote the projections of a configuration $(i, s, w)$ to $i$, $s$, and $w$, respectively.

For any $(s, a) \in (S_\exists \cup S_\forall) \times \Sigma$, let $T_{s,a} := \{(s,a,a',D,s') \in T \mid a' \in \Sigma,\ D \in \{-1,+1\},\ s' \in S\}$. A string $r = r_1r_2\ldots \in T^*\cup T^\omega$ is called a \emph{run} of $M$ on $w$ if for all $i$, $r_{i+1} \in T_{\pi_S(c_{i}), \pi_{\Sigma^*}(c_{i})(\pi_\N(c_{i}))}$, where $c_i = c_0\triangleright r_1\ldots r_i$ whenever $r_{i+1}$ exists. Any run $r$ is called an \emph{accepting run} if $r \in T^*$ and $\pi_S(c_0\triangleright r) = \sacc$. A strategy is a function $\sigma : \{(i, s, w, r) \in \N \times S \times \Sigma^* \times T^* \mid s \in S_\exists\} \rightarrow T$ such that $\sigma(i, s, w, r) \in T_{s, w_i}$ for any $i, s, w, r$. For a configuration $c$ and a run $r$, by abuse of notation, we will write $\sigma(c, r)$ to mean $\sigma(\pi_\N(c), \pi_S(c), \pi_{\Sigma^*}(c), r)$. A run is \emph{consistent} with a strategy $\sigma$ if for all $i$, $r_{i+1} = \sigma(c_0\triangleright r_1\ldots r_i, r_1\ldots r_i)$ whenever the latter is defined. The \emph{behaviour} of $\sigma$ is the set of runs consistent with $\sigma$.
A \emph{winning strategy} is a strategy $\sigma$ such that every run $r$ consistent with $\sigma$ is an accepting run. Note that whenever $\sigma$ is winning, its behaviour is finite.


Let $\sigma$ be a strategy, then we call $\sigma$ \emph{$N$-bounded} if for any $r$ in the behavior of $\sigma$, and any prefix $r_1\ldots r_k$ of $r$, $|\pi_{\Sigma^*}(c_0\triangleright r_1\ldots r_k)| < N$. A Turing machine is $N$-bounded if it has an $N$-bounded winning strategy.


%% file: app-NBA-0.tex
\subsection{Proof of \texorpdfstring{\cref{thm:NBA-0}}{Theorem \ref{thm:NBA-0}}} \label{app:NBA-0}
\thmNBAzero*
\begin{proof}
Given the main body, it remains to prove EXPTIME-hardness of $\P(\NBA) = 0$.

Recall that alternating PSPACE equals EXPTIME.
We give a polynomial-time reduction from the problem of acceptance of a word by a PSPACE-bounded alternating Turing machine.
Without loss of generality, we can assume that the Turing machine is linear-bounded, i.e., uses only the space occupied by the input word.

Let $M = (S_\exists, S_\forall, \Sigma, T, s_0, \sacc)$ be a linear-bounded alternating Turing machine.
Let $w = a_1 \cdots a_n \in \Sigma^*$ be the input word.
As mentioned before, we can assume that $M$ uses exactly $n$ tape cells.
We construct a BP~$\B$ and an NBA~$\A$ such that $\P_\B(\A \text{ accepts}) > 0$ if and only $M$ accepts~$w$.

The BP~$\B$ has the following set of types:
\[
 \Gamma \ = \ (\{1, \ldots, n\} \times S \times \Sigma) \ \cup\ (\{1, \ldots, n\} \times T) \ \cup\ (\{1, \ldots, n\} \times \{\chk\}) \ \cup\ \{E\}
\]
Intuitively, a type $(i,s,a)$ means that the head is at position~$i$, the current state is~$s$, and the head is reading letter~$a$; a type $(i,t)$ means that the head is at position~$i$, and transition~$t$ is being executed;
a type $(i,\chk)$ means that the accepting state has been reached, and cell~$i$ is being ``checked'' (in a sense to be explained later); type~$E$ indicates an error.
The type $(1,s_0,a_1)$ is the start type of~$\B$.

For all $(i,s,a) \in \Gamma$ with $s \in S_\exists$, include rules
\[
 (i,s,a) \btran{1/k} (i,t_j) \quad (1 \le j \le k)\,,
\]
where $\{t_1, \ldots, t_k\} = T_{s,a}$.
(Intuitively, a transition going out of an existential state is chosen as the only child, uniformly at random.)
For all $(i,s,a) \in \Gamma$ with $s \in S_\forall$, include a rule
\[
 (i,s,a) \btran{1} (i,t_1) \cdots (i,t_k) \,,
\]
where $\{t_1, \ldots, t_k\} = T_{s,a}$.
(Intuitively, all possible transitions going out of a universal state are children.)
For all $(i,(s,a,a',D,s')) \in \Gamma$ with $1 \le i{+}D \le n$, include rules
\[
 (i,(s,a,a',D,s')) \btran{1/|\Sigma|} (i{+}D, s', b_j) \quad (1 \le j \le |\Sigma|)\,,
\]
where $\{b_1, \ldots, b_{|\Sigma|}\} = \Sigma$.
(Intuitively, the letter in the cell at the new head position $i{+}D$ is guessed uniformly at random.)
For all $(i,(s,a,a',D,s')) \in \Gamma$ with $i{+}D \in \{0,n{+}1\}$, include a rule
\[
 (i,(s,a,a',D,s')) \btran{1} E\,.
\]
(Intuitively, when the space bound is exceeded, move to the error type~$E$.)
Include a rule $E \btran{1} E$ (i.e., a self-loop).
For all $(i, \sacc, a) \in \Gamma$, include a rule
\[
 (i, \sacc, a) \btran{1} (1,\chk) \cdots (n,\chk)\,.
\]
(Intuitively, after reaching~$\sacc$ all $n$ cells are ``checked''.)
For all $(i,\chk) \in \Gamma$, include a rule $(i,\chk) \btran{1} (i,\chk)$ (i.e., a self-loop).

The NBA $\A = (Q, \Gamma, \delta, Q_0, \{f\})$ has the following set of states:
\[
 Q \ = \ (\{1, \ldots, n\} \times \Sigma) \ \cup \ \{f\}
\]
The set of initial states is $Q_0 = \{(1,a_1), \ldots, (n,a_n)\}$ (recall that $a_1 \cdots a_n$ is the input word).
The idea is that if a prefix $X_1 \cdots X_k \in \Gamma^*$ of a tree branch corresponds to a prefix of a correct computation of~$M$, then the set of automaton states in $\delta(Q_0,X_1 \cdots X_k)$ corresponds to the tape after this computation prefix.
In fact, the transition relation $\delta \subseteq Q \times \Gamma \times Q$ is \emph{deterministic}, i.e., for all $q \in Q$ and $X \in \Gamma$ there is at most one $q'$ with $(q,X,q') \in \delta$.
Moreover, for any transition $((i,a), X, (j,a')) \in \delta$ we will have $i=j$.

For $(i,a) \in Q$ and all $(j,s,b) \in \Gamma$ with $i \ne j$ or $a=b$, include a self-loop
\[
 ((i,a),(j,s,b),(i,a)) \in \delta\,.
\]
(Intuitively, letter~$a$ in cell~$i$ is compatible with the head being on cell~$j$ and reading letter~$b$.)
For all $(i,a) \in Q$ and all $(j,t) \in \Gamma$ with $i \ne j$, include a self-loop
\[
 ((i,a),(j,t),(i,a)) \in \delta\,.
\]
(Intuitively, the content of cell~$i$ stays unchanged when the head is at position~$j$.)
For all $i,s,a,a',D,s'$ with $(i,a) \in Q$ and $(i,(s,a,a',D,s')) \in \Gamma$, include a transition
\[
 ((i,a), (i,(s,a,a',D,s')), (i,a'))\in \delta\,.
\]
(Intuitively, the transition $(s,a,a',D,s')$ changes the content of cell~$i$ from $a$ to~$a'$.)
For all $(i,a) \in Q$, include a transition
\[
 ((i,a),(i,\chk),f) \in \delta\,.
\]
(Intuitively, the type $(i,\chk)$ checks if the computation has been consistent in cell~$i$.)
For all $X \in \Gamma$, include a self-loop $(f,X,f) \in \delta$.

We will show that in this case, $\P_\B(\A \text{ accepts}) > 0$ if and only if $M$ accepts~$w$.
Firstly note that $f$ is a sink state in $\A$, and hence, if $\A$ reaches $f$ after reading some prefix of a branch, it will accept any branch with this prefix. This means that a tree $t$ is accepted by $\A$ if and only if there exists a finite prefix of $t$ such that $\A$ reaches $f$ on all of its branches.

Assume that $M$ accepts $w$. Then there exists a linear-bounded winning strategy $\sigma$ for the existential player. We will write $c_0$ for the initial configuration of $M$ on $w$, $R$ for the behaviour of $\sigma$, and $N$ for the length of the longest run in $R$. We will show that there exists a finite prefix generated by $\B$ with nonzero probability that precisely models this strategy, and that is such that $\A$ reaches $f$ on all of its branches. Since any tree with this prefix is accepted, this implies that $\P_\B(\A \text{ accepts}) > 0$ if $M$ accepts~$w$.

Let $t$ be the tree defined as follows:
\begin{itemize}
	\item the root of $t$ is $(1, s_0, a_1)$,
	\item for any node $(p_i, s_i, b_i)$ on level $2i$ with $s_i \in S_\exists$, let $t_1\ldots t_{i}$ be such that the branch from the root to $(p_i, s_i, b_i)$ has nodes $(p_j, t_{j+1})$ on level $2j+1$ for each $0 \leq j < i$. Then $(p_i, s_i, b_i)$ has a child $(p_i, \sigma(c_0 \triangleright t_1\ldots t_{i}, t_1\ldots t_i))$,
	\item for any node $(p_i, s_i, b_i)$ on level $2i$ with $s_i \in S_\forall$, let $t_1\ldots t_{i}$ be such that the branch from the root to $(p_i, s_i, b_i)$ has nodes $(p_j, t_{j+1})$ on level $2j+1$ for each $0 \leq j < i$. Let $w_i = \pi_{\Sigma^*}(c_0 \triangleright t_1\ldots t_{i})$. Then $(p_i, s_i, b_i)$ has children $(p_i, t_{i+1})$ for each $t_{i+1} \in T_{s_i, w_i(p_i)}$,
	\item for any node $(p_i, t_{i+1})$ with $t_{i+1} = (s, a, a', D, s')$ on level $2i+1$, let $t_1\ldots t_{i}$ be such that the branch from the root to $(p_i, t_{i+1})$ has nodes $(p_j, t_{j+1})$ on level $2j+1$ for each $0 \leq j < i$. Let $c_{i+1} = c_0 \triangleright t_1\ldots t_{i+1}$. Then $(p_i, t_{i+1})$ has a child $(p_i+D, s', \pi_{\Sigma^*}(c_{i+1})(p_i+D))$,
	\item any node $(p_i, \sacc, b_i)$ has children $(j, \chk)$ for each $1 \leq j \leq n$,
	\item and any node $(j, \chk)$ has a child $(j, \chk)$.
\end{itemize}
By construction, $t$ is now such that for any node $(p_i, t_{i+1})$ with $t_{i+1} = (s, a, a', D, s')$ on level $2i+1$, there exist runs in $R$ prefixed by $t_1 \ldots t_{i+1}$, where $t_1\ldots t_{i}$ are such that the branch from the root to $(p_i, t_{i+1})$ has nodes $(p_j, t_{j+1})$ on level $2j+1$ for each $0 \leq j < i$.
Hence, by the fact that $\sigma$ is linear-bounded, $1 \leq p_i+D \leq n$ and thus all the transitions in $t$ are according to the rules of $\B$.
Moreover, since the length of runs in $R$ is bounded by $N$, all the nodes at level $2N+1$ must be of the form $(i, \chk)$, and since all states $(i, \chk)$ are sink states, $t$ is generated with a nonzero probability.
Finally pick any state $(i, \chk)$ in $t$.
The only types from which $(i, a) \in Q$ does not have an outgoing edge are of the form $(i, s, b)$.
However, for any state $(i, s, b)$ at level $2j$ on the branch to $(i, \chk)$, let $t_1\ldots t_{j}$ be such that the branch from the root to $(i,s,b)$ has nodes $(p_k, t_{k+1})$ on level $2k+1$ for each $0 \leq k < j$.
Then $(i, b) = (i, \pi_{\Sigma^*}(c_0 \triangleright t_1\ldots t_{j})(i)) = \delta((i,w(i)), (p_0, t_1)\ldots(p_{j-1}, t_{j}))$ and hence $(i, a)$ survives, and reaches $f$ upon reading $(i, \chk)$.

For the other direction, assume that $\P_\B(\A \text{ accepts}) > 0$.
Then there exists a prefix $t$ generated by $\B$ with branches of length $N$ for some $N$ such that $\A$ reaches $f$ on all of its branches. W.l.o.g.~we can assume that $\pi_S(c_0) \neq \sacc$, because otherwise any strategy is winning.
Note that for any accepted tree $t$ and any branch prefix $(1, s_0, a_1)(1, t_1)\ldots (p_i, s_i, b_i)(p_i, t_{i+1})$ in $t$, the set of reachable states in the automaton from $Q_0$ reflects the tape contents of $M$, ie.~$\delta(Q_0, (1, s_0, a_1)(1, t_1)\ldots (p_i, s_i, b_i)(p_i, t_{i+1}))=\{(j, w(j))\mid 1 \leq j \leq n, w = \pi_{\Sigma^*}(c_0 \triangleright t_1\ldots t_{i+1})\}$. Also $s_i = \pi_S(c_0 \triangleright t_1\ldots t_i)$.
If this was not the case, then there would exist $j$ such that $\delta((j, a_j),  (1, s_0, a_1)(1, t_1)\ldots (p_i, s_i, b_i)(p_i, t_{i+1})) = \emptyset$ and hence the branch reaching $(j, \chk)$ prefixed by $(1, s_0, a_1)(1, t_1)\ldots (p_i, s_i, b_i)(p_i, t_{i+1})$ does not reach $f$.
Let $n_0n_1\ldots$ be any branch in $t$.
Pick $m$ such that $n_{2i+1} = (a_i, t_{i+1})$ for all $0 \leq i < m$ and $n_{2m+1} = (a_m, t_{m+1})$.
Let $c = c_0 \triangleright t_1\ldots t_m$.
Then we define $\sigma$ to be any strategy such that for any such $c$, if $\pi_S(c) \in S_\exists$, then $\sigma(c, t_1\ldots t_m) = t_{m+1}$.
This is a valid strategy since for any branch in $t$, the set of reachable states of the automaton after reading a branch prefix reflects the tape contents of $M$ and the alphabet characters contained in the nodes of the branch prefix have to reflect the automaton states.
Note that $\sigma$ is $n$-bounded. We claim that $\sigma$ is a winning strategy and hence $M$ accepts $w$.

Let $r = t_1t_2\ldots$ be a run consistent with $\sigma$. We will show that there exists a branch $(1, s_0, a_1)(1, t_1)\ldots (p_i, t_{i+1})(p_{i+1}, \sacc, w_{i+1}(p_{i+1}))$ in $t$ with $p_i = \pi_\N(c_0\triangleright t_1\ldots t_i)$, $s_i = \pi_S(c_0\triangleright t_1\ldots t_i)$, and $w_i = \pi_{\Sigma^*}(c_0\triangleright t_1\ldots t_i)$ (and hence $|r| = i+1$ and $\pi_S(c_0\triangleright r) = \sacc$).
\begin{itemize}
	\item For any branch prefix $(1, s_0, a_1)(1, t_1)\ldots (p_j, t_{j+1})$ where $t_{j+1} = (s, a, a', D, s')$, according to the rules of $\B$,  $(p_j, t_{j+1})$ has a single child $(p_i+D, s', b)$. Let $c_{j+1} = c_0 \triangleright t_1\ldots t_{j+1}$. Since trees prefixed by $t$ are accepted, $p_i+D = \pi_\N(c_{j+1})$, $s' = \pi_S(c_{j+1})$, and $b = \pi_{\Sigma^*}(c_{j+1})(\pi_\N(c_{j+1}))$.
	\item For any branch prefix $(1, s_0, a_1)(1, t_1)\ldots (p_{j-1}, t_{j})(p_j, s_j, b_j)$ where $s_j \in S_\forall$, according to the rules of $\B$, $(p_j, s_j, b_j)$ has children $(p_j, t')$ for each $t' \in T_{s_j, b_j}$. Let $c_j = c_0\triangleright t_1\ldots t_j$. Since trees prefixed by $t$ are accepted, $s_j =  \pi_S(c_j)$ and $b_j = \pi_{\Sigma^*}(c_j)(\pi_\N(c_j))$. Hence, $t_{j+1} \in T_{s_j, b_j}$ and $(p_j, s_j, b_j)$ has a child $(p_j, t_{j+1})$.
	\item For any branch prefix $(1, s_0, a_1)(1, t_1)\ldots (p_{j-1}, t_{j})(p_j, s_j, b_j)$ where $s_j \in S_\exists$, according to the rules of $\B$, $(p_j, s_j, b_j)$ has a single child $(p_j, t')$. By definition of $\sigma$, $\sigma(c_0\triangleright t_1\ldots t_j, t_1\ldots t_j) = t'$.
\end{itemize}
Since $t$ is finite, every branch reaches states of the form $(p, \sacc, b)$ in a finite number of steps. Hence, $r$ is finite and reaches $\sacc$ and since $r$ is any run consistent with $\sigma$, $\sigma$ is an winning strategy. Thus, $M$ accepts $w$.

\end{proof}

%% file: app-coNBA-1.tex
\subsection{Proof of \texorpdfstring{\cref{lem:UBA-X1-f}}{Lemma \ref{lem:UBA-X1-f}}} \label{app:coNBA-1}
\lemUBAXonef*
\begin{proof}
Consider a branch $X_0 X_1 X_2 \cdots$ of a $\B$-tree accepted by~$\A$.
Then $\A$ has an accepting run \[q_0 \xrightarrow{X_0} q_1 \xrightarrow{X_1} q_2 \xrightarrow{X_2} \cdots\] with $q_0 \in Q_0$.
By the pigeonhole principle, there are $f \in F$ and $X_f \in \Gamma$ such that this run contains the segment $f \xrightarrow{X_f}$ infinitely often.
By its construction, the B\"uchi automaton $\A \times \B$ has the accepting run \[(q_0,X_0) \xrightarrow{X_1} (q_1,X_1) \xrightarrow{X_2} \cdots\,,\] which contains the state $(f,X_f)$ infinitely often.
Let $k \ge 1$ be such that $(q_k,X_k) = (f,X_f)$.
Then $\A[f,X_f]$ accepts $X_k X_{k+1} \cdots$ via the run \[\bar{q}_0 \xrightarrow{X_k} (q_k,X_k) \xrightarrow{X_{k+1}} (q_{k+1},X_{k+1}) \cdots\,.\]
Therefore, denoting by $E(f,X_f,n)$ for $n \in \N$ the event that there exists a branch $X_k X_{k+1} \cdots$ emanating from the $n$th (in a breadth-first order) node in the tree (necessarily a node of type $X_k = X_f$) such that $\A[f,X_f]$ has an accepting run \[\bar{q}_0 \xrightarrow{X_k} (q_k,X_k) \xrightarrow{X_{k+1}} (q_{k+1},X_{k+1}) \cdots\,,\] we have
\begin{align*}
\P_{X_0}(\A \text{ accepts some branch})\ \le \ &\sum_{f \in F} \sum_{X_f \in \Gamma} \sum_{n \in \N} \P_{X_0}(E(f,X_f,n))\,.
\end{align*}
Further, 
\begin{align*}
  &\ \P_{X_0}(E(f,X_f,n)) \\
 =&\ \P_{X_0}(\text{the $n$th node has type~$X_f$})  \cdot \P_{X_f}(\A[f,X_f] \text{ accepts some branch}) \\
 \le&\ \P_{X_f}(\A[f,X_f] \text{ accepts some branch})\,.
\end{align*}
The ``only if'' direction follows.

Towards the ``if'' direction, suppose that $\A \times \B$ has a path $(q_0,X_0) \Xrightarrow{w}* (f,X_f)$ with $q_0 \in Q_0$ and $f \in F$ such that the probability that some branch of a $\B[X_f]$-tree is accepted by~$\A[f,X_f]$ is positive.
Then there is a successor, say~$X_f'$, of~$X_f$ such that the probability that some branch of a $\B[X_f']$-tree is accepted by $\A \times \B$ when started in $(f,X_f)$ is positive.
Thus, the probability that some branch of a $\B$-tree (starts with $X_0 w X_f'$ and) is accepted by~$\A$ is positive.
\end{proof}

%% file: app-coNBA-0.tex
\subsection{Proof of \texorpdfstring{\cref{thm:coNBA-0}}{Theorem~\ref{thm:coNBA-0}}} \label{app:coNBA-0}
\thmcoNBAzero*
\begin{proof}
Towards membership in EXPTIME, an NBA can be translated, in exponential time, to a DPA of exponential size; see, e.g., \cite{Piterman07}.
By shifting the priorities (colours) in the DPA by~$1$, we can make the DPA accept exactly those words that are rejected by the NBA.
Since $\P(\DPA) = 0$ is in P by \cref{thm:DPA-0}, it follows that $\P(\coNBA) = 0$ is in EXPTIME.

Concerning EXPTIME-hardness, we adapt the construction of the proof of \cref{thm:NBA-0}.
As in that proof, let $M = (S_\exists, S_\forall, \Sigma, T, s_0, \sacc)$ be a linear-bounded alternating Turing machine.
Let $w = a_1 \cdots a_n \in \Sigma^*$ be the input word, and assume again that $M$ uses exactly $n$ tape cells.
We construct a BP~$\B$ and an NBA~$\A$ such that the probability that all branches of the random tree are rejected by~$\A$ is positive if and only $M$ accepts~$w$.

For~$\B$ we use almost the same construction as in \cref{thm:NBA-0}, except that we do not need the types $(1,\chk), \ldots, (n,\chk)$.
We replace them by a single type $\en$.
Accordingly, for all $(i, \sacc, a) \in \Gamma$, we replace the rule
\begin{align*}
 (i, \sacc, a) &\btran{1} (1,\chk) \cdots (n,\chk)
\intertext{by a rule}
 (i, \sacc, a) &\btran{1} \en\,,
\end{align*}
and include a rule $\en \btran{1} \en$ (i.e., a self-loop).

We want to construct the NBA~$\A$ so that it accepts exactly those branches that correspond to infinite computations that do not arrive at $\sacc$, or to ``non-computations'', i.e., where the ``guessing'' in a rule
\[
 (i,(s,a,a',D,s')) \btran{1/|\Sigma|} (i{+}D, s', b_j) \quad (1 \le j \le |\Sigma|)\,,
\]
has been wrong.

The NBA $\A = (Q, \Gamma, \delta, Q_0, Q)$ has the same set of states as in \cref{thm:NBA-0}:
\[
 Q \ = \ (\{1, \ldots, n\} \times \Sigma) \ \cup \ \{f\}
\]
As in \cref{thm:NBA-0}, the set of initial states is $Q_0 = \{(1,a_1), \ldots, (n,a_n)\}$ (recall that $a_1 \cdots a_n$ is the input word).
As in \cref{thm:NBA-0}, the idea is that if a prefix $X_1 \cdots X_k \in \Gamma^*$ of a tree branch corresponds to a prefix of a correct computation of~$M$, then the set of automaton states in $\delta(Q_0,X_1 \cdots X_k)$ corresponds to the tape after this computation prefix.
In fact, the transition relation $\delta \subseteq Q \times \Gamma \times Q$ is \emph{deterministic}, i.e., for all $q \in Q$ and $X \in \Gamma$ there is at most one $q'$ with $(q,X,q') \in \delta$.
Moreover, for any transition $((i,a), X, (j,a')) \in \delta$ we will have $i=j$.
Unlike in \cref{thm:NBA-0}, all states are accepting.

For $(i,a) \in Q$ and all $(j,s,b) \in \Gamma$ with $i \ne j$ or $a=b$, include a self-loop
\[
 ((i,a),(j,s,b),(i,a)) \in \delta\,.
\]
(Intuitively, letter~$a$ in cell~$i$ is compatible with the head being on cell~$j$ and reading letter~$b$.)
For $(i,a) \in Q$ and all $(i,s,b) \in \Gamma$ with $a \ne b$, include a transition
\[
 ((i,a),(i,s,b),f) \in \delta\,.
\]
(Intuitively, letter~$a$ in cell~$i$ is \emph{not} compatible with the head being on cell~$i$ and reading letter~$b$; i.e., the prefix of the branch does not correspond to a correct computation.)
For all $(i,a) \in Q$ and all $(j,t) \in \Gamma$ with $i \ne j$, include a self-loop
\[
 ((i,a),(j,t),(i,a)) \in \delta\,.
\]
(Intuitively, the content of cell~$i$ stays unchanged when the head is at position~$j$.)
For all $i,s,a,a',D,s'$ with $(i,a) \in Q$ and $(i,(s,a,a',D,s')) \in \Gamma$, include a transition
\[
 ((i,a), (i,(s,a,a',D,s')), (i,a'))\in \delta\,.
\]
(Intuitively, the transition $(s,a,a',D,s')$ changes the content of cell~$i$ from $a$ to~$a'$.)
For all $X \in \Gamma$, include a self-loop $(f,X,f) \in \delta$.
Note that $(f,\en,f)$ is the only transition labeled with~$\en$.

In this way:
\begin{itemize}
\item
If a tree branch does not correspond to a correct computation, the automaton~$\A$ enters the state~$f$, remains there forever, and, thus, accepts.
\item
If a tree branch corresponds to an infinite computation not entering~$\sacc$, the set of states that~$\A$ can be in always reflects the tape. Thus, $\A$ accepts.
\item
If a tree branch corresponds to a computation entering~$\sacc$, the branch also enters~$\en$, and $\A$ does not enter~$f$. Thus, $\A$ rejects.
\end{itemize}
It follows that the probability that all branches of the random tree are rejected by~$\A$ is positive if and only $M$ accepts~$w$.
A more detailed argument would follow very similar lines as the proof of \cref{thm:NBA-0}.
\end{proof}

%% file: app-LTL-0.tex
\section{Proof of \texorpdfstring{\cref{thm:LTL-0}}{Theorem \ref{thm:LTL-0}}} \label{app:LTL-0}
\thmLTLzero*
\begin{proof}
Given the main body, it remains to show 2EXPTIME-hardness of $\P(\LTL) = 0$. Our construction is inspired by the proof of 2EXPTIME-hardness of model-checking a concurrent probabilistic program (another name for MDP) against an LTL formula, see Theorem 3.2.1 in \cite{CourcoubetisYannakakis95}.

We will use the fact that 2EXPTIME is equal to alternating EXPSPACE. Let $M = (S_\exists, S_\forall, \Sigma, T, s_0, \sacc)$ be an alternating Turing machine whose work tape usage is bounded by $2^n$ on any input of length $n$. Without loss of generality, we can assume that the machine has two possible next moves for each configuration and that it halts when it reaches the accepting state $\sacc$. For a given alternating TM $M$ and an input $w$ of length $n$, we will construct a BP~$\B$ and an LTL formula $\varphi$ both of size $\mathcal{O}(n)$ such that $\P_\B(\varphi) > 0$ if and only if $M$ accepts~$w$.

The branching process $\B$ is defined by the diagram in Fig.~\ref{fig:bp}. Every node in the diagram has a unique label that corresponds to a type of $\B$, although not all nodes are explicitly labelled. The $\CIRCLE$-nodes represent randomising branching and $\Square$-nodes represent tree branching. Namely, if a $\CIRCLE$-node $x$ has $k$ successors $y_1,\ldots, y_k$:  $\vcenter{\xymatrix@1@R=10pt@C=10pt{ & *+[l]{x\,\CIRCLE} \ar[dr] \ar[dl] &\\y_1 &  \cdots &y_k}}$, then $\B$ has the rules $x \btran{1/k} y_i$ for $i=1,\dots,k$. On the other hand, if a $\Square$-node $x$ has $k$ successors $y_1,\ldots, y_k$:  $\vcenter{\xymatrix@1@R=10pt@C=10pt{ & *+[l]{x\,\Square} \ar[dr] \ar[dl] &\\y_1 &  \cdots &y_k}}$, then $\B$ has the rule $x \btran{1} y_1\cdots y_k$.

The start type of $\B$ is $a$. The detailed diagram of the blocks $I$, $O$, $N$, $D_1$, $D_2$ and $D_3$ is shown in Fig.~\ref{fig:bl} on the left. All nodes in these blocks are $\CIRCLE$-nodes. The diagram of the blocks $C_1$ and $C_2$ is shown in Fig.~\ref{fig:bl} on the right. These blocks have $\Square$-nodes in the first $n+1$ levels, and the rest are $\CIRCLE$-nodes. The initial and final nodes are labelled by $u$ and $v$, respectively. The nodes below $u$ are labelled with $(\ell_i,0)$ and $(\ell_i,1)$, $i=1,\ldots,n$, as shown in the picture. Every block has its own unique set of labels $u$, $v$, $(\ell_i,0)$ and $(\ell_i,1)$, $i=1,\ldots,n$; however we do not distinguish them in the diagram for simplicity.

In every block, except for $I$, the nodes above $v$ are labelled by $(\ell_{n+1},\delta)$, where $\delta \in \Sigma\;\cup\; (S_\exists\cup S_\forall)\times \Sigma$ ranges over the symbols of the extended work tape alphabet. In block $I$, there are $n+1$ nodes above $v$ which are labelled by $(\ell_{n+1},\gamma_i)$ for $i=1,\ldots,n+1$ such that for an input word $w=w_1\ldots w_n$ we have $\gamma_1=(s_0,w_1)$, $\gamma_i=w_i$ for $1<i\leq n$, and $\gamma_{n+1}=\Box$ is the blank symbol. These nodes will define the initial configuration of the Turing machine $M$.

The intended behaviour of process $\B$ is as follows. It starts generating a tree $T$ with a root node $a$. Then it tries to constructs the initial configuration of the Turing machine $M$ on input $w=w_1\ldots w_n$ by looping through block $I$ for $2^n$ many times. Each iteration of the block $I$ produces a string of the form $u(\ell_1,b_1)\ldots (\ell_n,b_n)(\ell_{n+1},z)v$, where $b_1\ldots b_n$ is the address of a work tape cell in binary, and $z$ is the content of that cell. $\B$ is supposed to construct the initial configuration by specifying the content of the work tape starting with $0$ and ending with cell $2^n-1$.

Each iteration of the loop, closed by the arc $c\rightarrow b$, corresponds to a move from one configuration of the Turing machine to the next. First, in block $O$, process $\B$ tries to reproduce the current configuration (in the first iteration of the loop, it is the initial configuration) by making $2^n$ iterations. Then, using tree branching in block $C_1$, the process produces a full binary tree of height $n$, each branch of which looks like $u(\ell_1,b_1)\ldots (\ell_n,b_n)$. Note that the last type $(\ell_n,b_n)$ is randomising, and after it $\B$ tries to correctly reproduce the content of the cell $b_1\ldots b_n$ in the current configuration. If the current state of $M$ is existential, then $\B$ is expected to move to $m_1$; if it is universal, it is expected to move to $m_2$. Two successors of $m_1$ and $m_2$ correspond to two possible moves out of the current configuration. In $m_1$, $\B$ randomly chooses the next move; in $m_2$, $\B$ makes a tree branching with two children corresponding to two possible next moves. Then, in block $N$, $\B$ tries to reproduce the next configuration of $M$ in the same way as it produced the current configuration in block $O$. In block $C_2$, the process produces a full binary tree of height $n$ (in the same way as it does it in block $C_1$), and after each branch of the form $u(\ell_1,b_1)\ldots (\ell_n,b_n)$ it tries to correctly reproduce the content of the cell $b_1\ldots b_n$ in the next configuration. Then, on every branch $u(\ell_1,b_1)\ldots (\ell_n,b_n)(\ell_{n+1},z)v$ produced by $C_2$, in blocks $D_1$, $D_2$, $D_3$, the process tries to reproduce the content of the cell $b_1\ldots b_n$ from the old configuration (block $O$) and its two adjacent cells. Finally, in state $c$, $\B$ is expected to move from $c$ to the sink state $d$ if the new configuration if accepting. Otherwise, $\B$ is expected to move to $b$.

We now define an LTL formula $\varphi$ that describes the expected behaviour of process $\B$. The formula $\varphi$ is the conjunction of the following parts:
\begin{enumerate}
	\item In blocks $I$, $O$ and $N$, the process constructs a configuration of $M$ cell-by-cell in order starting from cell $0$ and ending with cell $2^n-1$.
	\item In block $I$, the process constructs the initial configuration.
	\item On a branch produced by $C_1$ that corresponds to index $k$, the cell content specified by $C_1$ is equal to that of cell $k$ defined in block $O$, in block $I$ (if this is the first iteration of block $C_1$) and in the previous iteration of block $N$ (if there was any).
	\item If the current configuration if existential, then $\B$ moves to $m_1$. Otherwise, it moves to $m_2$.
	\item On a branch produced by $C_2$ that corresponds to index $k$, the cell content specified by $C_2$ is equal to that of cell $k$ defined in block $N$, and the indices in blocks $D_1$, $D_2$, $D_3$ are $k-1$, $k$, $k+1$, respectively.
	\item The cell contents in blocks $D_1$, $D_2$, $D_3$ are equal to those defined in block $O$.
	\item The cell content specified by $C_2$ follows directly from the contents of the cells specified by $D_1$, $D_2$, $D_3$ and the rule of $M$ that was chosen on the current branch.
	\item If the new configuration is accepting, then $\B$ moves from $c$ to $d$. Otherwise, it moves to $b$.
	\item $\mathsf{F}\mathsf{G}\, d$, that is, eventually $d$ always holds.
\end{enumerate}

The above properties can be expressed using LTL formulas. We will not explicitly write them down but the details of these formulas are very similar to those defined in the proof of Theorem 3.2.1 from \cite{CourcoubetisYannakakis95}.

Now suppose that Turing machine $M$ accepts an input $w$. Hence there is a $2^n$-bounded winning strategy for the existential player. In this case, with some positive probability, $\B$ can generate a tree $T$ all whose branches satisfy $\varphi$ as follows:
\begin{itemize}
	\item First, it generates the initial configuration in blocks $I$ and $O$.
	\item On every branch produced by $C_1$, process $\B$ generates the correct cell content.
	\item Then it moves to either $m_1$, if the current configuration is existential, or to $m_2$, if it is universal.
	\item In the former case, $\B$ chooses the next move that agrees with the winning strategy of the existential player.
	\item Then in block $N$, it generates a new configuration of $M$ that follows from $O$ according to the chosen move. (If the tree branching node $m_2$ was chosen, then $\B$ generates correct new configurations on every branch.)
	\item Next, $\B$ generates correct cell content on each branch produced by $C_2$ and chooses correct indices and cell contents in blocks $D_1$, $D_2$, $D_3$.
	\item Finally, $\B$ moves from $c$ to $d$, if an accepting configuration is reached, or it moves back to $b$ otherwise. In the latter case, $\B$ generates in block $O$ the same configuration that was generated in $N$ and continues the process.
\end{itemize}
Note that since the existential player has a winning strategy, state $d$ will eventually appear on every branch of $T$. Since $T$ is finitely branching, it follows by K\"onig's lemma that the above process will reach $d$ on every branch of $T$ in a finite number of steps. After that, $\B$ repeats the rule $d \btran{1} d$ forever. Clearly, all these events can happen with some positive probability. Hence, $\P_\B(\varphi) > 0$.

Conversely, suppose there is a positive probability that $\B$ generates a tree $T$ all whose branches satisfy $\varphi$. Hence every branch of $T$ eventually reaches state $d$ (and then always repeats it). Since $T$ is finitely branching, it follows by K\"onig's lemma that there exists a finite prefix of $T$ whose every leaf is labelled by $d$. This prefix encodes the moves of the existential player (after each appearance of state $m_1$ in the prefix) that allow him to reach the accepting configuration, no matter what the universal player does. In other words, the existential player has a winning strategy, and hence $M$ accepts $w$. Formally, this can be shown along similar lines as in the proof of \cref{thm:NBA-0}.

Therefore, we proved that $M$ accepts $w$ if and only if $\P_\B(\varphi) > 0$.
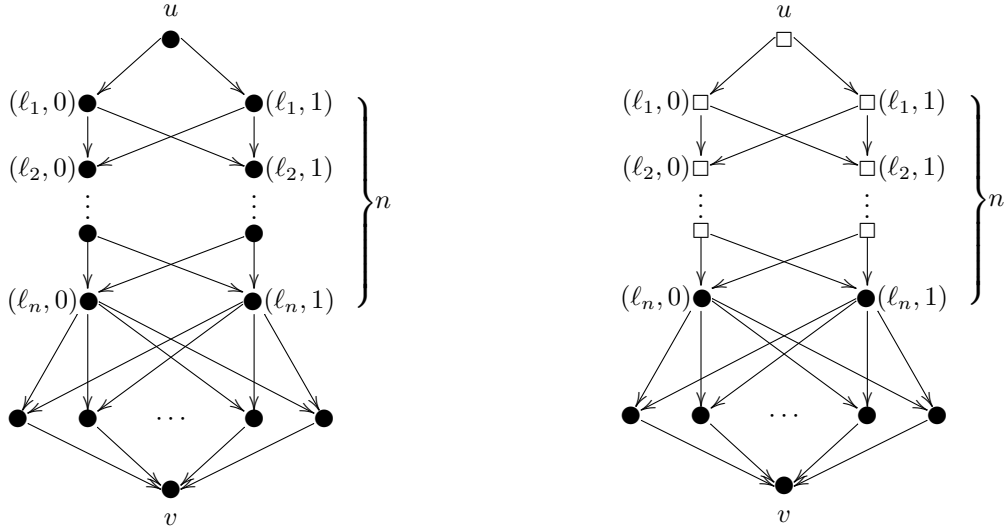
\begin{figure}[hp]
	\centerline{
	\xy
	\xymatrix "M"@R=15pt@C=5pt{
		& & *!<0pt,-1pt>{\txt{$u$\\ \CIRCLE}} \ar[dl]\ar[dr] & & \\
		& *!<13pt,0pt>{(\ell_1,0)\CIRCLE} \ar[d]\ar[drr] & & *!<-13pt,0pt>{\CIRCLE(\ell_1,1)} \ar[d]\ar[dll] & \\
		& *!<13pt,0pt>{(\ell_2,0)\CIRCLE} \ar@{}[d] |<<<{\vdots} &  & *!<-13pt,0pt>{\CIRCLE(\ell_2,1)} \ar@{}[d] |<<<{\vdots} & \\
		& *{\CIRCLE} \ar[d]\ar[drr] & & *{\CIRCLE} \ar[d]\ar[dll] & \\
		& *!<13pt,2pt>{(\ell_n,0)\CIRCLE} \ar[ddl] \ar[dd] \ar[ddrr] \ar[ddrrr]& & *!<-13pt,2pt>{\CIRCLE(\ell_n,1)} \ar[ddr] \ar[dd] \ar[ddll] \ar[ddlll] & \\
		&&&&\\
		*{\CIRCLE} \ar[drr] & *{\CIRCLE} \ar[dr] & {\cdots} & *{\CIRCLE} \ar[dl] & *{\CIRCLE} \ar[dll]\\
		& & *!<0pt,3pt>{\txt{\CIRCLE\\ $v$}} & &
	}
	\POS"M2,4"."M5,4"!C*-!L\frm{\}},!C*++++++!L\txt{$n$}
	\POS+(40,23)
	\xymatrix "N"@R=15pt@C=5pt{
	& & *!<0pt,-1pt>{\txt{$u$\\ \Square}} \ar[dl]\ar[dr] & & \\
	& *!<13pt,0pt>{(\ell_1,0)\Square} \ar[d]\ar[drr] & & *!<-13pt,0pt>{\Square(\ell_1,1)} \ar[d]\ar[dll] & \\
	& *!<13pt,0pt>{(\ell_2,0)\Square} \ar@{}[d] |<<<{\vdots} &  & *!<-13pt,0pt>{\Square(\ell_2,1)} \ar@{}[d] |<<<{\vdots} & \\
	& *{\Square} \ar[d]\ar[drr] & & *{\Square} \ar[d]\ar[dll] & \\
	& *!<13pt,2pt>{(\ell_n,0)\CIRCLE} \ar[ddl] \ar[dd] \ar[ddrr] \ar[ddrrr]& & *!<-13pt,2pt>{\CIRCLE(\ell_n,1)} \ar[ddr] \ar[dd] \ar[ddll] \ar[ddlll] & \\
	&&&&\\
	*{\CIRCLE} \ar[drr] & *{\CIRCLE} \ar[dr] & {\cdots} & *{\CIRCLE} \ar[dl] & *{\CIRCLE} \ar[dll]\\
	& & *!<0pt,3pt>{\txt{\CIRCLE\\ $v$}} & &
	}
	\POS"N2,4"."N5,4"!C*-!L\frm{\}},!C*++++++!L\txt{$n$}
	\endxy
	}
	\caption{Diagrams of the randomising (left) and tree (right) branching blocks.}
	\label{fig:bl}
\end{figure}
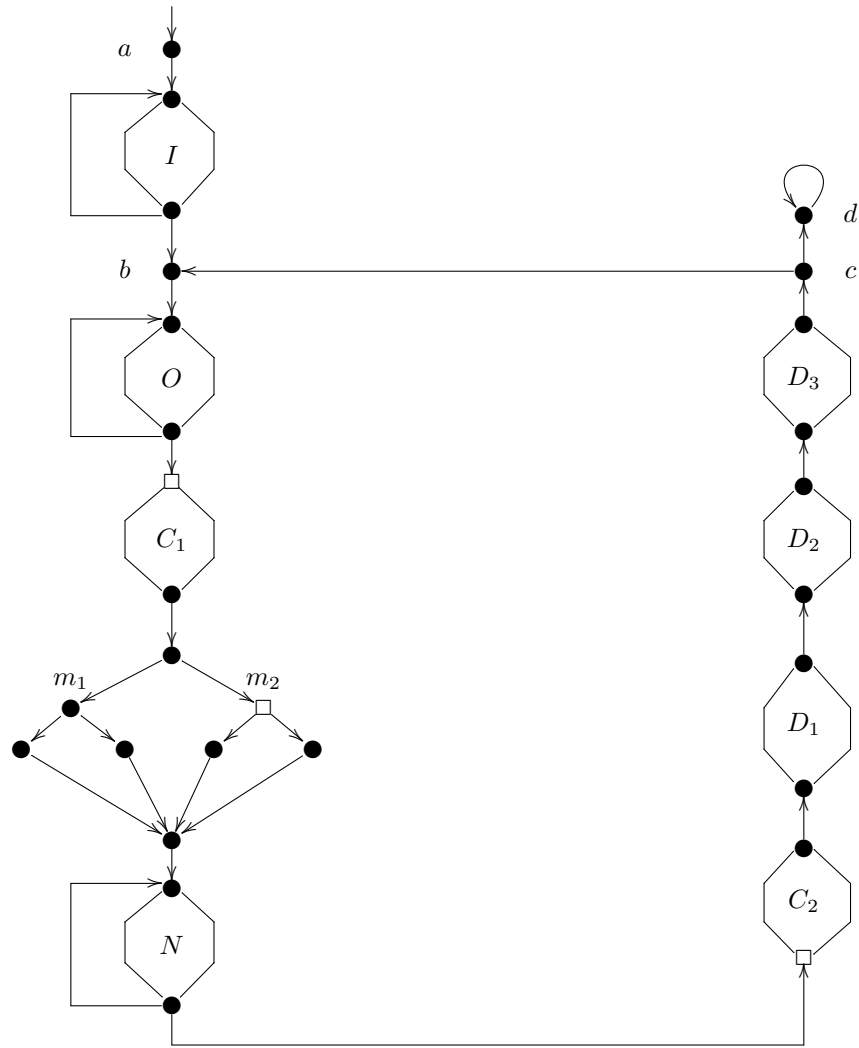
\begin{figure}[hp]
	\centerline{
	\xymatrix @R=8pt@C=8pt{
		&  &  & *+0{} \ar[d] & & & &\\
		&  & a & *{\CIRCLE} \ar[d] & & & &\\
		& *+0{} \ar[rr] \ar@{-}[ddd] &  & *!<0pt,2pt>{\CIRCLE} \ar@{}[ddd] |*+{I} \ar@{-}[dl] \ar@{-}[dr] & & & &\\
		&  & *+0{} \ar@{-}[d] &  & *+0{} \ar@{-}[d] & & &\\
		&  & *+0{} \ar@{-}[dr] & & *+0{} \ar@{-}[dl] & & &\\
		& *+0{} \ar@{-}[rr] &  & *!<0pt,-2pt>{\CIRCLE} \ar[d] & & & &&&&&&&&&&&&&*{\CIRCLE}\ar@(ur,ul)[] &d\\
		&  & b & *{\CIRCLE} \ar[d] & & & & *+0{}  &&&&&&&&&&&&*{\CIRCLE}\ar[u] \ar[llllllllllllllll]& c\\
		& *+0{} \ar[rr] \ar@{-}[ddd] &  & *!<0pt,2pt>{\CIRCLE} \ar@{}[ddd] |*+{O} \ar@{-}[dl] \ar@{-}[dr] & & & &&&&&&&&&&&&&*!<0pt,2pt>{\CIRCLE} \ar[u]\\
		&  & *+0{} \ar@{-}[d] &  & *+0{} \ar@{-}[d] & & &&&&&&&&&&&& *+0{} \ar@{-}[ur] & & *+0{} \ar@{-}[ul]\\
		&  & *+0{} \ar@{-}[dr] & & *+0{} \ar@{-}[dl] & & &&&&&&&&&&&&*+0{} \ar@{-}[u] &  & *+0{} \ar@{-}[u]\\
		& *+0{} \ar@{-}[rr] &  & *!<0pt,-2pt>{\CIRCLE} \ar[d] &  & & &&&&&&&&&&&&& *!<0pt,-2pt>{\CIRCLE} \ar@{}[uuu] |*+{D_3} \ar@{-}[ul] \ar@{-}[ur]\\
		&  &  & *{\Square} \ar@{}[ddd] |*+{C_1} \ar@{-}[dl] \ar@{-}[dr] & & & &&&&&&&&&&&&&*!<0pt,2pt>{\CIRCLE} \ar[u]\\
		&  & *+0{} \ar@{-}[d] &  & *+0{} \ar@{-}[d] & & &&&&&&&&&&&&*+0{} \ar@{-}[ur] & & *+0{} \ar@{-}[ul]\\
		&  & *+0{} \ar@{-}[dr] & & *+0{} \ar@{-}[dl] & & &&&&&&&&&&&&*+0{} \ar@{-}[u] &  & *+0{} \ar@{-}[u]\\
		&  &  & *!<0pt,-2pt>{\CIRCLE} \ar[d] & & & &&&&&&&&&&&&&*!<0pt,-2pt>{\CIRCLE} \ar@{}[uuu] |*+{D_2} \ar@{-}[ul] \ar@{-}[ur]\\
		& *!<0pt,10pt>{m_1} &  & *{\CIRCLE} \ar[dll] \ar[drr]&  & *!<0pt,10pt>{m_2} & &&&&&&&&&&&&&*!<0pt,3pt>{\CIRCLE} \ar[u]\\
		& *{\CIRCLE} \ar[dl] \ar[dr] &  &  &  & *{\Square} \ar[dl] \ar[dr] & &&&&&&&&&&&&*+0{} \ar@{-}[ur] & & *+0{} \ar@{-}[ul]\\
		*{\CIRCLE} \ar[ddrrr] &  & *{\CIRCLE}\ar[ddr] &  & *{\CIRCLE}\ar[ddl] & &*{\CIRCLE}\ar[ddlll] &&&&&&&&&&&&*+0{} \ar@{-}[u] &  & *+0{} \ar@{-}[u]\\
		&  &  & &  & & &&&&&&&&&&&&&*!<0pt,-2pt>{\CIRCLE} \ar@{}[uuu] |*+{D_1} \ar@{-}[ul] \ar@{-}[ur]\\
		&  &  & *{\CIRCLE} \ar[d] &  & & &&&&&&&&&&&&&*!<0pt,3pt>{\CIRCLE} \ar[u]\\
		& *+0{} \ar[rr] \ar@{-}[ddd] &  & *!<0pt,2pt>{\CIRCLE} \ar@{}[ddd] |*+{N} \ar@{-}[dl] \ar@{-}[dr] & & & &&&&&&&&&&&&*+0{} \ar@{-}[ur] & & *+0{} \ar@{-}[ul]\\
		&  & *+0{} \ar@{-}[d] &  & *+0{} \ar@{-}[d] & & &&&&&&&&&&&&*+0{} \ar@{-}[u] &  & *+0{} \ar@{-}[u]\\
		&  & *+0{} \ar@{-}[dr] & & *+0{} \ar@{-}[dl] & & &&&&&&&&&&&&&*!<0pt,-2pt>{\Square} \ar@{}[uuu] |*+{C_2} \ar@{-}[ul] \ar@{-}[ur]\\
		& *+0{} \ar@{-}[rr] &  & *{\CIRCLE} \ar@{-}[d]& & & &&&&&&&&&&&&&\\
		&  &  & *+0{} \ar@{-}[rrrrrrrrrrrrrrrr] & & & &&&&&&&&&&&&&*+0{} \ar[uu]\\
	}
	}
	\caption{Diagram of the branching process $\B$.}
	\label{fig:bp}
\end{figure}
\end{proof}